\documentclass{article}

\usepackage[utf8]{inputenc} 
\usepackage[T1]{fontenc}    
\usepackage{hyperref}       
\usepackage{url}            
\usepackage{booktabs}       
\usepackage{amsfonts, amsmath, amsthm}       
\usepackage{nicefrac}       
\usepackage[expansion=false]{microtype} 
\usepackage{xcolor, paralist}         
\usepackage{multirow, multicol}
\usepackage{pdflscape}
\usepackage{longtable, float}
\usepackage{graphicx}
\usepackage{geometry, subcaption}
\usepackage{cleveref}
\floatstyle{plain}
\restylefloat{table}

\newtheorem{theorem}{Theorem}
\newtheorem{lemma}[theorem]{Lemma}

\theoremstyle{definition}

\theoremstyle{remark}

\title{EigenLI: Spectral Approximations to Late Interaction}

\author{
Archish S\textsuperscript{1}\thanks{Work done while at Microsoft Research India.},
Sabyasachi Basu\textsuperscript{2}\thanks{Corresponding author.},
Ankit Garg\textsuperscript{2},
Ravishankar Krishnaswamy\textsuperscript{2},\\
Kirankumar Shiragur\textsuperscript{2}
\\
\footnotesize
\texttt{archish17@gmail.com,\{sabyasachi.basu,garga,rakri,kshiragur\}@microsoft.com}
\\
\small
\textsuperscript{1}Temple
\quad
\textsuperscript{2}Microsoft Research India
}
\begin{document}

\maketitle

\begin{abstract}
  Late-interaction models such as ColBERT achieve strong effectiveness by representing each document with many token-level vectors, but this expressivity leads to large indexing cost, storage footprints and expensive MaxSim scoring. We show that late-interaction representations exhibit an intrinsic low-rank structure: document token embeddings concentrate in a low-dimensional subspace that preserves most of the retrieval signal. Leveraging this observation, we introduce EigenLI, a spectral approximation framework that compresses late-interaction representations via document-specific low-dimensional subspaces. Unlike clustering or pooling methods, EigenLI identifies the dominant eigendirections of each document and uses them to construct reduced interaction representations. 
Empirically, $k$-EigenLI with $k \le 32$ outperforms k-means and Ward clustering based pooling methods on ColBERTv2 and AnswerAI-ColBERT-small; GTE-ModernColBERT exhibits a different tradeoff at $k=32$, where clustering methods perform better. The same spectral construction also yields EigenLI-SV, an ANN-compatible single-vector representation derived from the second-order summary of the reduced structure. Across multiple datasets and all three text models, EigenLI-SV consistently outperforms comparable single-vector surrogates such as MUVERA. 

\end{abstract}

\section{Introduction}

Multi-vector retrieval models \cite{khattab2020colbert} (also known as late-interaction models) represent queries and documents as sets of token-level vectors, with relevance score computed using the MaxSim function between the two sets of vectors. These models have yielded significant gains in retrieval quality over single-vector models which represent queries and documents as single vectors and compute relevance score using cosine similarity or dot product. Multi-vector models of much smaller size (in terms of number of model parameters) are able to beat single-vector models of much larger size in a number of domains (e.g. Reason-ModernColBERT on BrowseComp-Plus \cite{lighton2026bloatedretriever}, also see \cite{takehi2025fantastic}). Multi-vector models have been shown to being immune to catastrophic forgetting~\cite{SAGKS26} and generally believed to generalize better to OOD data (e.g. generalization to new languages in \cite{SCMC26}). Moreover, a recent line of work has provably established the fundamental representational limits of single vectors, and the advantages of multi-vector methods~\cite{weller2025, j26, KIZM26, jayaram2026near, agarwal2026retrieval}. However, multiple vectors per query and document lead to increased indexing, storage and latency costs. A lot of work has gone into improving the footprint of multi-vector retrieval systems (via better indexing and quantization strategies) including \cite{santhanam2022colbertv2, santhanam2022plaid, lee2023rethinking, scheerer2025warp} and also into standardizing and designing efficient training strategies for multi-vector models \cite{chaffin2025pylate, chaffin2026colbert, clavie2025simple}. Despite this, the concerns of large indexing and storage costs for multi-vector retrieval systems remain. To further improve efficiency, it is therefore crucial to reduce the number of vectors used in representing a document.

Multi-vector models are also playing a crucial role in the Visual Document Retrieval domain where queries are usually text but on the document side we have images/screenshots\cite{faysse2024colpali, moreira2026nemotron, kolavi2025m3dr, gunther2025jina}. Due to patch level tokenization, these models represent documents using an even larger number of vectors (usually around 1000 per document). This further motivates the question of reducing the representation size (in terms of number of vectors) of multi-vector retrieval models. 

Strategies for compressing multi-vector representations can be broadly categorized into two parts: ones that require training and ones that don't. Among approaches that require training include \cite{macavaney2025efficient, xiao2025metaembed, qin2026multi}. In contrast, training-free strategies, primarily focus on pruning and pooling strategies \cite{lassance2021study, acquavia2023static, he2025token, clavie2024reducing, yan2025docpruner, yan2026sculpting, yan2026visual}. Among these, k-means and hierarchical (Ward linkage) clustering based pooling approaches are known to perform best among all the pruning and pooling approaches \cite{clavie2024reducing, jha2026brief}.

\begin{figure}
    \centering
    \includegraphics[width = 0.9\textwidth]{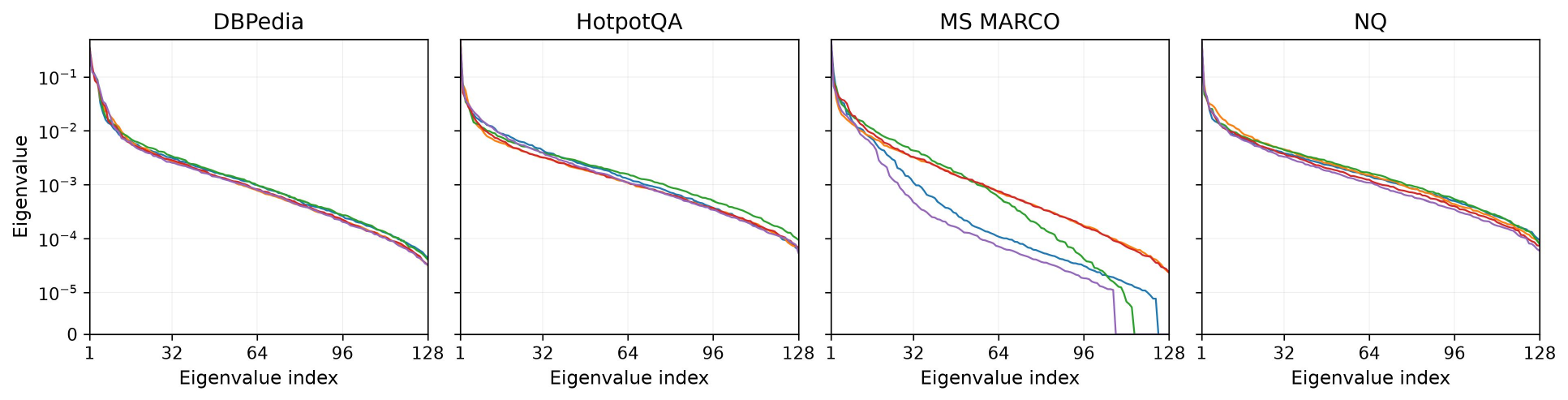}
    \includegraphics[width = 0.9\textwidth]{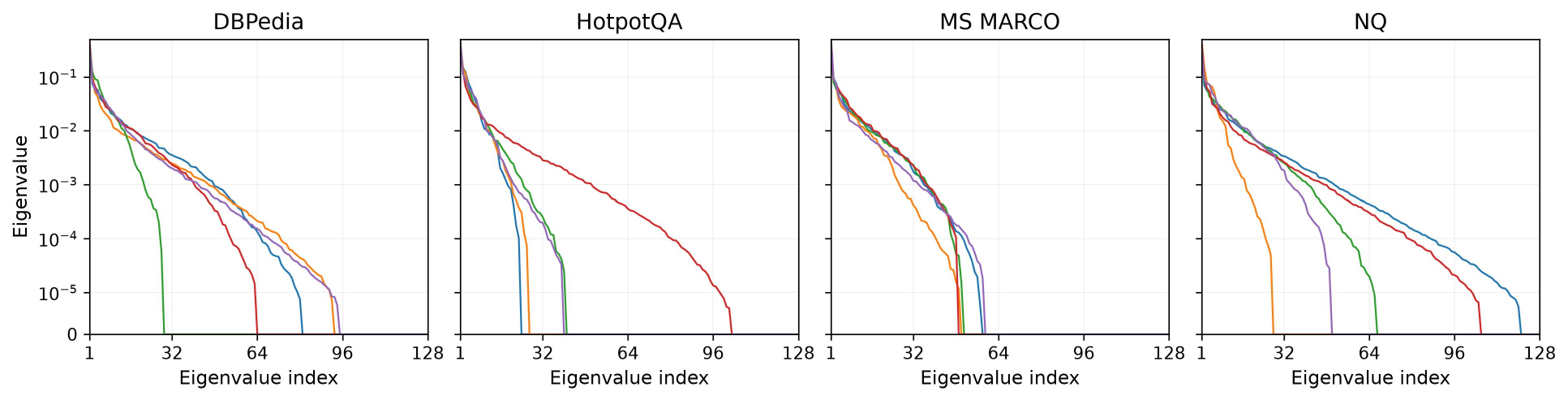}
    \caption{The spectra of the second moment matrix (normalized) of five longest (top) and randomly chosen (bottom) documents for four datasets. Each line refers to a different document.}\label{fig:eigens}
\end{figure}

\subsection{Our Contributions}

Our main contribution is a new training-free approach for reducing the number of vectors for representing a document together with a new scoring function different from MaxSim, which we call \emph{EigenLI}, named so because we use the spectral properties of late-interaction representations. 

\subsubsection{$k$-EigenLI}

The starting point of the new compression strategy is the empirical observation that in multi-vector retrieval models, the vectors in the multi-vector embedding of a document approximately lie in a low-dimensional subspace. Suppose $(d_1,\ldots, d_m)$ are the vectors representing a document embedded by a multi-vector retrieval model such as ColBERTv2. We observe that the matrix $\sum_i d_i d_i^T$ is approximately low rank and has rapidly decaying eigenvalues, as seen in Figure~\ref{fig:eigens}. This motivates the following compression strategy. We compute the top-$k$ eigenvectors $w_1,\ldots, w_k$ of the document second-moment matrix $\sum_i d_i d_i^T$ and use them to represent the document. Given the multi-vector embedding of a query $(q_1, \ldots, q_n)$, we define the score $s := \sum_{i, j} \langle q_i, w_j \rangle^2$. 

Note that we do not compute MaxSim between the query vectors and the eigenvectors of the second moment matrix. Instead, we compute how much of the mass of each query token-vector projects onto the top eigenspace of the document second-moment matrix. There are alternative formulations of the scoring function which motivate it even better. Suppose $\Pi$ is the projector onto the subspace spanned by $w_1,\ldots, w_k$. Then $s = \sum_i \|\Pi q_i\|_2^2$. Alternatively let $S_D$ be the subspace spanned by $w_1, \ldots, w_k$. Then $s = \sum_i \max_{v\in S_D, \|v\|_2 = 1} \langle q_i, v \rangle^2$. This relates our new scoring function to MaxSim (see also Section \ref{sec:subspace_maxsim}). A novel aspect of our compression strategy is that we compress the set-of-vectors representation of a document to a subspace-of-vectors representation whereas all the previous compression strategies try to compress the set-of-vectors representation to a smaller sized set-of-vectors representation. 

\subsubsection{EigenLI-SV}

Another salient feature of our compression strategy and scoring function is that it can be represented exactly using a high-dimensional single-vector dot product. Note that $s = \sum_{i, j} \langle q_i, w_j \rangle^2 = \langle \sum_i q_i \otimes q_i, \sum_j w_j \otimes w_j \rangle$. This gives a $d^2$-dimensional single-vector encoding ($d$ is the token-vector dimension of the multi-vector model). We can further use the quadratic kernel $K(x) := (x_1^2,\ldots, x_n^2, \ldots, \sqrt{2} x_i x_j, \ldots)$ to reduce the dimension to $d(d+1)/2$. When the token-vector dimension is $d = 128$ (e.g., for ColBERTv2), we get an $8256$-dimensional single-vector encoding. Note that the dimension of the single vector is independent of $k$, the number of eigenvectors we choose. We compare performance against $10240$-dimensional FDEs produced by MUVERA \cite{dhulipala2024MUVERA}. For ColBERTv2 on the 13 BEIR development sets, EigenLI-SV at $k = 32$ improves nDCG@10 over MUVERA by $69.6\%$ under the geometric mean and $78.9\%$ under the arithmetic mean of the 13 per-dataset performance ratios, expressed as a relative improvement. The EigenLI-SV dimension is $8256$, smaller than the MUVERA encoding dimension of $10240$. Single-vector encodings can use highly optimized approximate nearest-neighbor systems \cite{jegou2011product,ge2013optimized,malkov2020hnsw,johnson2019faiss,subramanya2019diskann,guo2020scann,chen2021spann,singh2021freshdiskann}.

\subsubsection{Model and Domain Coverage}

We evaluate text models from multiple families on BEIR and a Qwen-based visual model on ViDoRe-v3. The results are model dependent. For ColBERTv2 on BEIR at $k=32$, the arithmetic mean over the 13 per-dataset relative improvements in nDCG@10 is $15.0\%$ over $k$-means++ and $8.0\%$ over Ward pooling. AnswerAI-ColBERT-small shows the same qualitative trend, whereas GTE-ModernColBERT favors clustering at $k=32$ on several metrics. For ColQwen3 4B on ViDoRe-v3 at $k=32$, the arithmetic mean over the 8 per-dataset relative improvements in nDCG@10 is $5.5\%$ over $k$-means++ and $3.4\%$ over Ward. Performance can decline when $k$ becomes a large fraction of the token-vector dimension, as the selected subspaces become less discriminative. We observe this at $k=64$ for ColBERTv2 (token dimension 128) and $k=128$ for ColQwen3 (token dimension 320). These results motivate model-aware selection of $k$, which we discuss in \Cref{sec:discussion}.

\section{Related Work}

We survey the various compression strategies for reducing the number of vectors in multi-vector representations. \cite{lassance2021study, acquavia2023static} explored IDF based pruning, attention based pruning and other simple heuristics. \cite{he2025token} proposed a learned mechanism for pruning unimportant token vectors. \cite{clavie2024reducing} proposed clustering based pooling strategies for reducing number of vectors and empirically showed that hierarchical clustering (with Ward linkage) works better than k-means on average across BEIR datasets (this is something we see in our experiments too). \cite{yan2025docpruner} introduced an attention-score based pruning mechanism for Visual Document Retrieval (VDR) multi-vector models. \cite{yan2026sculpting} introduced a two-stage prune-then-merge framework for improving the performance of merging (pooling) based compression techniques for VDR models. \cite{yan2026visual} add a positional embedding component to the final patch-level embeddings in VDR multi-vector models and then apply hierarchical clustering based pooling. \cite{jha2026brief} survey the various training-free compression strategies and find that k-means and Ward clustering based pooling methods work the best.

Among the compression methods that require training, \cite{macavaney2025efficient} introduced a learned linear pooling layer that learns to pool multiple vectors into smaller number of vectors. \cite{xiao2025metaembed} explored the use meta/memory tokens whose output vectors act as summary vectors for the document and also introduced a Matryoshka type loss that helps the model learn multi-vector representations with varying number of vectors. \cite{qin2026multi} propose a method that uses meta/memory token vectors to assign attention scores and find the important token vectors for a document and cluster and pool around these important token vectors. They show that their method outperforms the previous two training-based compressed strategies. \cite{josef2026learn, chaffin2026hpoolregularization} explore learnt kmeans/Ward based pooling strategies.

Principal Component Analysis (PCA) is of course a well known technique that uses singular value decomposition (SVD) for dimensionality reduction. In the context of retrieval and embeddings, it has been used to reduce the dimension of single-vector embeddings \cite{siciliano2024static}. The natural use of PCA for compressing multi-vector embeddings would be to reduce the dimension of each token vector (project onto principal components). However, we use this technique to reduce the number of vectors representing a document by focusing on the top-$k$ eigenspace of the token-vector second moment matrix and representing the document using this eigenspace. The fact that this can be done on a per document basis is surprising. Moreover, in multi-vector compression, compression alone is not sufficient and we need to see how queries interact with the compressed representations to model relevance. We introduce a new simple scoring function as well. 

There is literature on representing documents as subspaces (e.g. \cite{piwowarski2010filtering, zuccon2009semantic, caputo2011query, piwowarski2009quantum}) but the use cases, methodology and techniques are very different from ours.

\section{Method: $k$-EigenLI}

Our primary contribution is \textbf{$k$-EigenLI}, a spectral method for reducing the cost of late-interaction retrieval while preserving most of its effectiveness. The key idea is to replace the original token-vector representation of each document with a document-specific low-dimensional subspace.

\paragraph{Setup.}
Let a document be represented by token vectors
$$
D = \{d_1,\dots,d_m\} \subset \mathbb{R}^d,
$$
and query be represented by token vectors
$$
Q = \{q_1,\dots,q_n\} \subset \mathbb{R}^d.
$$
As usual, we assume all the token vectors have $L_2$ norm $1$.

\paragraph{Document-specific spectral subspace.}
Our starting point is that the token embeddings of a document often lie close to a lower-dimensional subspace of the original embedding space (see Figure ~\ref{fig:eigens}). Rather than compressing the document by clustering tokens into a small number of representatives, we compute the principal eigenvectors of the second moment matrix of the document token vectors. Take $M_D := \sum_{i=1}^m d_i d_i^T$ and let $w_1,\ldots, w_k$ denote the top eigenvectors of $M_D$ (corresponding to the highest $k$ eigenvalues). We define the score between a query and document as follows:
$$
s_k(Q,D) := \sum_{i = 1}^n \sum_{j = 1}^k \langle q_i, w_j\rangle^2
$$

The following lemma gives some equivalent formulations of the above scoring function.

\begin{lemma}\label{lem:scorealt}
Let $w_1, \ldots, w_k \in \mathbb{R}^d$ be a set of orthonormal vectors and let $\Pi := \sum_{i=1}^k w_i w_i^T \in \mathbb{R}^{d \times d}$ be the projector onto the subspace spanned by $w_1,\ldots, w_k$. Also let $S$ be the subspace of $\mathbb{R}^d$ spanned by $w_1, \ldots, w_k$. Then for any vectors $q_1, \ldots, q_n$
$$
\sum_{i = 1}^n \sum_{j = 1}^k \langle q_i, w_j\rangle^2 = \sum_{i=1}^n \| \Pi q_i\|_2^2 = \sum_{i=1}^n \max_{v\in S, \| v\|_2 = 1} \langle q_i, v \rangle^2
$$
\end{lemma}

\begin{proof}
For any vector $q$, 
\begin{align*}
    \sum_{j = 1}^k \langle q, w_j \rangle^2 &= \sum_{j = 1}^k q^T w_j w_j^T q \\
    &=  q^T \left(\sum_{j = 1}^k w_j w_j^T \right) q\\
    &= q^T \Pi  q \\
    &= \|\Pi q\|_2^2
\end{align*}
For the second equality, first consider any vector $v$ s.t. $v \in S$ and $\|v\|_2 = 1$. Then
$$
\langle q, v \rangle = \langle \Pi q, v \rangle + \langle (I - \Pi) q, v \rangle
$$
Since $v \in S$ and $(I - \Pi) q$ lies in a subspace orthogonal to $S$, $\langle (I - \Pi) q, v \rangle = 0$. Hence $\langle q, v \rangle = \langle \Pi q, v \rangle$. Now by Cauchy-Schwarz,
$$
\langle \Pi q, v \rangle^2 \le \|\Pi q\|_2^2 \|v\|_2^2 = \|\Pi q\|_2^2
$$
To see equality, set $v = \Pi q/\|\Pi q\|_2$. With this value of $v$, 
$$
\langle q, v \rangle^2 = \langle q, \Pi q \rangle^2/|\Pi q\|_2^2 = |\Pi q\|_2^4/|\Pi q\|_2^2 = |\Pi q\|_2^2
$$
This completes the proof.
\end{proof}

Most modern approaches rely on clustering tokens vectors to the desired budget; our approach presents a systematic alternative that leverages the properties exhibited by multi-vector representations. Moreover, our approach exhibits substantial gains compared to the best known multi-vector compression strategies, which we demonstrate over a wide range of text and visual document retrieval datasets in Section \ref{sec:experiments}.

\subsection{EigenLI-SV: A Single-Vector Derivative of k-EigenLI}

A salient feature of our scoring function and compression technique is that, the same spectral construction also yields a strong single-vector representation. This inherently makes for an important bridge with existing retrieval systems, that are typically graph based approximate nearest neighbor search algorithms that run on single vector embeddings. We refer to this variant as \textbf{EigenLI-SV}.

Note that our scoring function is 
$$
s_k(Q,D) = \sum_{i = 1}^n \sum_{j = 1}^k \langle q_i, w_j\rangle^2
$$
where $(q_1, \ldots, q_n)$ is the multi-vector embedding of a query and $w_1, \ldots, w_k$ are the top $k$ eigenvectors of $\sum_{i = 1}^m d_i d_i^T$, where $(d_1,\ldots, d_m)$ is the multi-vector embedding of a document. We can rewrite the scoring function as 
$$
s_k(Q,D) = \left\langle \sum_{i=1}^n q_i \otimes q_i, \sum_{j=1}^k w_j \otimes w_j \right\rangle
$$
This immediately gives a $d^2$-dimensional single-vector encoding ($d$ is the dimension of each token vector) for queries and documents such that the dot product captures our scoring function. We can reduce the dimension by a factor of $2$ using the quadratic polynomial kernel. Let $K: \mathbb{R}^d \rightarrow \mathbb{R}^{d(d+1)/2}$ be the map $K(x) := (x_1^2,\ldots, x_n^2, \ldots, \sqrt{2} x_i x_j, \ldots)$. Then we can see that
$$
s_k(Q,D) = \left\langle \sum_{i=1}^n K(q_i), \sum_{j=1}^k K(w_j) \right\rangle
$$
This gives a $d(d+1)/2$-dimensional single-vector encoding (independent of $k$) which captures our scoring function. This is useful for models which have token-vector embedding dimension $d = 128$, as is the case for ColBERTv2. Here we get an $8256$-dimensional single-vector encoding. While the embedding dimension is large, using the same data types, this is still almost $8$ times smaller than a representation that uses 512 128-dimensional vectors. The encoding is mathematically score-equivalent to the $k$-EigenLI compressed multivector from which it is derived; in our implementation, the resulting metrics agree within a numerical tolerance of $10^{-3}$. Moreover, such single vector representations still compare quite favorably against the indexing and storage costs of multi-vector models as also seen in the case of MUVERA \cite{dhulipala2024MUVERA}. Note that the $32$-EigenLI representation contributes a total of 4096 dimensions across all vectors, so the single-vector representation is approximately twice as large in raw dimensionality. However they are compatible with ANN systems and dot-product computation has much lower computation cost as opposed to Maxsim or our scoring function for EigenLI which require many pairwise dot products. In Section \ref{sec:experiments}, we compare EigenLI-SV against the $10240$-dimensional FDEs of MUVERA.

\subsection{Relating MaxSim to our score function}\label{sec:subspace_maxsim}

Here we provide some mathematical intuition as to how our score function is related to MaxSim. For a fixed query and document, let the query multi-vector embedding be $(q_1, \ldots, q_n)$ and document multi-vector embedding be $(d_1, \ldots, d_m)$. Let $w_1, \ldots, w_k$ be the top $k$-eigenvectors of $\sum_{i=1}^m d_i d_i^T$. Let $\lambda_1, \ldots, \lambda_d$ be the eigenvalues of $\sum_{i=1}^m d_i d_i^T$ (in descending order). Let $\Pi_k$ denote the projection onto the subspace spanned by $w_1, \ldots, w_k$. Then we have the following lemma:
\begin{lemma}\label{lem:pca}
$$
\sum_{i=1}^m \|d_i - \Pi_k d_i\|_2^2 = \sum_{r \ge k+1} \lambda_r
$$
In particular, if $(1- \epsilon)$-fraction of the mass is concentrated in the top-$k$ eigenvalues, i.e. $\sum_{r = 1}^k \lambda_r \ge (1-\epsilon) \left(\sum_{r=1}^d \lambda_r\right)$, then
$$
\frac{1}{m} \sum_{i=1}^m \|d_i - \Pi_k d_i\|_2^2 \le \epsilon
$$
\end{lemma}

\begin{proof}
\begin{align*}
\sum_{i=1}^m \|d_i - \Pi_k d_i\|_2^2 &= \sum_{i=1}^m \text{tr} \left[ (I - \Pi_k) d_i d_i^T (I - \Pi_k)\right]\\
&= \text{tr} \left[ (I - \Pi_k) \left(\sum_{i=1}^m d_i d_i^T\right) (I - \Pi_k)\right] \\
&= \sum_{r \ge k+1} \lambda_r
\end{align*}
The third equality follows from the fact that $\Pi_k$ is projector onto the top-$k$ eigenvectors of $\sum_{i=1}^m d_i d_i^T$. Since all the $d_i$'s are unit norm, it follows that 
$$
\sum_{r = 1}^d \lambda_r = \text{tr} \left( \sum_{i=1}^m d_i d_i^T \right) = m
$$
Hence 
$$
\frac{1}{m} \sum_{i=1}^m \|d_i - \Pi_k d_i\|_2^2 = \frac{\sum_{r \ge k+1} \lambda_r}{\sum_{r=1}^d \lambda_r} \le \epsilon
$$
\end{proof}
This shows that, on average, document token embeddings lie close to the top-$k$ subspace, which explains why the subspace-based scoring performs extremely well in practice. We now relate a variant of our score function with MaxSim. Note that our score function is
$$
s_k(Q,D) = \sum_{i = 1}^n \sum_{j = 1}^k \langle q_i, w_j\rangle^2
$$
and the MaxSim scoring function is
$$
\sum_{i = 1}^n \max_{j = 1}^m \langle q_i, d_j \rangle
$$
By Lemma \ref{lem:scorealt}, we can write our scoring function as
$$
s_k(Q,D) = \sum_{i = 1}^n \max_{v \in S, \|v\|_2 = 1} \langle q_i, v \rangle^2
$$
where $S$ is the subspace spanned by $w_1, \ldots, w_k$. Consider a close variant where
$$
s'_k(Q,D) = \sum_{i = 1}^n \max_{v \in S, \|v\|_2 = 1} \langle q_i, v \rangle
$$
Note that if $(1-\epsilon)$-fraction of the mass is concentrated in top-$k$ eigenvalues, then by Lemma \ref{lem:pca}, $\frac{1}{m} \sum_{i=1}^m \|d_i - \Pi_k d_i\|_2^2 \le \epsilon$. This leaves the room for some $d_i$'s to have large deviations, but for this analysis, we would need a stronger assumption that for all $i$, $\|d_i - \Pi_k d_i\|_2\le \delta$. Now let $q_i$ is mapped to document token-vector $d_{\tau(i)}$ under MaxSim. Then
\begin{align*}
    s'_k(Q,D) &= \sum_{i = 1}^n \max_{v \in S, \|v\|_2 = 1} \langle q_i, v \rangle \\
              &\ge \sum_{i = 1}^n \langle q_i, \Pi_k d_{\tau(i)}/\|\Pi_k d_{\tau(i)}\|_2 \rangle \\
              &\ge \sum_{i = 1}^n \langle q_i, \Pi_k d_{\tau(i)} \rangle\\
              &\ge \sum_{i = 1}^n \langle q_i, d_{\tau(i)}\rangle - n\delta \\
              &= \sum_{i = 1}^n \max_{j = 1}^m \langle q_i, d_j \rangle - n \delta
\end{align*}
First inequality follows from the fact that $\Pi_k$ projects any vector into the subspace $S$. Second inequality follows from the fact that $\Pi_k$ is a projection matrix and hence $\|\Pi_k d_{\tau(i)}\|_2 \le 1$. Third inequality follows from the fact that $\|d_\tau(i) - \Pi_k d_\tau(i)\|_2\le \delta$. And the final equality follows from the definition of $\tau$.

This provides some mathematical intuition as to why our score function approximates the score function. If a document is assigned high score under maxsim, it would also be assigned high score using our scoring function. In the other direction, we expect that for an irrelevant document, the corresponding document subspace is unrelated to the query token vectors and hence would get assigned a low score.

An important point about our scoring function is that we don't use the eigenvalues. An alternative scoring function is $\sum_{i = 1}^n \sum_{j = 1}^k \lambda_j \langle q_i, w_j \rangle^2$. However, in our preliminary experiments, this approach consistently underperformed our existing subspace-based technique. The reason is the following
\begin{align*}
    \sum_{i = 1}^n \sum_{j = 1}^k \lambda_j \langle q_i, w_j \rangle^2 = \sum_{i = 1}^n q_i^T \left( \sum_{j=1}^k \lambda_j w_j w_j^T \right) q_i
\end{align*}
$\sum_{j=1}^k \lambda_j w_j w_j^T$ is a low rank approximation of $\sum_{\ell=1}^m d_\ell d_\ell^T$ and hence 
\begin{align*}
    \sum_{i = 1}^n q_i^T \left( \sum_{j=1}^k \lambda_j w_j w_j^T \right) q_i \sim  \sum_{i = 1}^n q_i^T \left( \sum_{\ell=1}^m d_\ell d_\ell^T \right) q_i = \sum_{i = 1}^n \sum_{\ell = 1}^m \langle q_i, d_\ell \rangle^2 
\end{align*}
And the above scoring function performs poorly compared to maxsim since it is summing up the similarities between all query and document token-vectors. Hence dropping the eigenvalues is important to closely mimic the MaxSim scoring function.

\section{Experiments}\label{sec:experiments}

In this section, we evaluate the performance of EigenLI as a drop-in compression technique for late-interaction retrieval models. We perform extensive evaluation across MS MARCO passages, BEIR datasets, and the ViDoRe-v3 multimodal benchmark. On BEIR, we evaluate ColBERTv2, AnswerAI-ColBERT-small, and GTE-ModernColBERT; the main text gives model-level summaries, and the appendix gives detailed per-dataset results for all three models. ColBERTv2 documents are represented using 512 vectors of 128 dimensions each, while queries use 32 vectors. For the ViDoRe-v3 benchmark, we use ColQwen3 VL 4B model where documents use $\sim 1250$ vectors of 320 dimensions each, while queries use 45 vectors.


We report truncated Recall, nDCG, and MRR at cutoffs 10, 100, and 1000 across all BEIR datasets. For ViDoRe-v3, we report Recall and nDCG only at @10 and @100 due to the relatively smaller dataset size. 

\paragraph{Implementation details:} 

We use the publicly available ColBERTv2, AnswerAI-ColBERT-small, and GTE-ModernColBERT weights from Hugging Face. For ColQwen, we use the publicly available Hugging Face checkpoint TomoroAI/tomoro-colqwen3-embed-4b, a 4B-parameter ColQwen3 multimodal late-interaction model. The reproduced BEIR sweep uses four B200 GPUs, with one task per GPU; the ViDoRe-v3 evaluations use a 4$\times$A100 GPU instance. ColBERTv2 documents are represented using 512 vectors. These vectors include the padding token vectors and punctuation token vectors. Padding token vectors are zeroed out as is the convention for ColBERTv2. So some of these 512 vectors could be 0 if the number of tokens in the document is less than 512. If the number of tokens in the document are above 512, then the document is truncated appropriately so that the total number of vectors is 512. Consistently, query expansion is used in the models.

\paragraph{Baselines: } In the multi-vector setting, our experiments compare EigenLI against two corresponding token-vector compression baselines: pooling using $k$-means++ and Ward. We emphasize that we intend to compare only with lightweight, training free compression strategies. We report results at per-document vector counts of $8,16,32$ and include uncompressed MaxSim as a reference rather than a strict upper bound, since EigenLI uses a different scoring function. In the single-vector setting, we compare EigenLI-SV to MUVERA FDE on the BEIR datasets for all three text models. Note that the dimension of the single vector produced by our method scales as $d(d+1)/2$; for the ColQwen3 4B model, this would result in a dimension of 51,360, which is prohibitive. We therefore do not run the single-vector comparison on ColQwen3. For ColBERTv2, the resulting EigenLI-SV dimension is 8,256.

\paragraph{Performance of $k$-EigenLI: } On ColBERTv2 and AnswerAI-ColBERT-small, EigenLI-32 improves over both $k$-means++ and Ward across Recall, nDCG, and MRR. GTE-ModernColBERT behaves differently: at $k=32$, its EigenLI representation is competitive with $k$-means++ at deeper Recall cutoffs but trails both clustering baselines on nDCG and MRR. This distinction is visible in \Cref{tab:beir_geomean_32}, which reports all three metrics at all three thresholds for every model. The full ColBERTv2 Recall, nDCG, and MRR results are reported in \Cref{tab:beir_recall,tab:beir_ndcg,tab:beir_mrr}; the ViDoRe-v3 Recall and nDCG results are reported in \Cref{tab:vidore_recall,tab:vidore_ndcg}. \Cref{tab:sv_recall} gives the per-dataset EigenLI-SV/MUVERA Recall comparison. Complete per-dataset results for AnswerAI-ColBERT-small and GTE-ModernColBERT, including EigenLI-SV and all baselines, are in \Cref{tab:answerai_all,tab:gte_all} in the appendix.

\begin{table}[ht]
\centering
\caption{Geometric (GM) and arithmetic (AM) means over the 13 per-dataset performance ratios, expressed as a relative improvement, of EigenLI-32 over each baseline on BEIR, grouped by embedding model.}
\label{tab:beir_geomean_32}
\scriptsize
\resizebox{\textwidth}{!}{%
\begin{tabular}{lllrrrrrrrrr}
\toprule
Model & Baseline & Mean & R@10 & R@100 & R@1000 & nDCG@10 & nDCG@100 & nDCG@1000 & MRR@10 & MRR@100 & MRR@1000 \\
\midrule
ColBERTv2 & KMeans++ & GM & +12.9\% & +8.2\% & +4.6\% & +13.8\% & +11.4\% & +9.6\% & +12.5\% & +11.8\% & +11.7\% \\
 & KMeans++ & AM & +13.8\% & +8.8\% & +5.0\% & +15.0\% & +12.2\% & +10.2\% & +13.4\% & +12.5\% & +12.5\% \\
& Ward & GM & +6.7\% & +3.5\% & +2.0\% & +7.6\% & +5.9\% & +4.9\% & +7.3\% & +7.0\% & +7.0\% \\
& Ward & AM & +7.0\% & +3.7\% & +2.2\% & +8.0\% & +6.2\% & +5.0\% & +7.6\% & +7.3\% & +7.3\% \\
& Full MaxSim & GM & -3.0\% & -3.2\% & -2.6\% & -3.3\% & -3.9\% & -4.1\% & -2.3\% & -2.5\% & -2.4\% \\
& Full MaxSim & AM & -2.9\% & -3.1\% & -2.5\% & -3.2\% & -3.8\% & -4.0\% & -2.1\% & -2.2\% & -2.2\% \\
\midrule
AnswerAI-ColBERT-small & KMeans++ & GM & +10.2\% & +9.4\% & +6.9\% & +11.0\% & +11.0\% & +9.7\% & +8.1\% & +7.6\% & +7.6\% \\
& KMeans++ & AM & +10.9\% & +9.9\% & +7.3\% & +11.9\% & +11.7\% & +10.2\% & +8.8\% & +8.2\% & +8.2\% \\
& Ward & GM & +2.4\% & +4.5\% & +3.4\% & +2.2\% & +3.5\% & +2.9\% & +0.6\% & +0.5\% & +0.5\% \\
& Ward & AM & +2.7\% & +4.8\% & +3.6\% & +2.6\% & +4.0\% & +3.2\% & +0.9\% & +0.8\% & +0.8\% \\
& Full MaxSim & GM & -12.2\% & -10.3\% & -6.8\% & -13.6\% & -13.2\% & -11.9\% & -12.0\% & -11.8\% & -11.8\% \\
& Full MaxSim & AM & -11.9\% & -10.0\% & -6.6\% & -13.3\% & -13.0\% & -11.7\% & -11.6\% & -11.5\% & -11.4\% \\
\midrule
GTE-ModernColBERT & KMeans++ & GM & -2.5\% & +0.7\% & +1.6\% & -2.8\% & -1.0\% & -0.5\% & -2.6\% & -2.6\% & -2.6\% \\
& KMeans++ & AM & -1.6\% & +1.0\% & +2.1\% & -1.9\% & -0.4\% & +0.1\% & -2.0\% & -2.0\% & -2.0\% \\
& Ward & GM & -8.9\% & -3.9\% & -0.8\% & -10.5\% & -7.9\% & -6.5\% & -9.0\% & -8.8\% & -8.8\% \\
& Ward & AM & -8.6\% & -3.7\% & -0.7\% & -10.2\% & -7.7\% & -6.3\% & -8.7\% & -8.6\% & -8.5\% \\
& Full MaxSim & GM & -12.5\% & -7.9\% & -4.8\% & -15.4\% & -13.4\% & -12.1\% & -14.3\% & -14.0\% & -14.0\% \\
& Full MaxSim & AM & -11.9\% & -7.6\% & -4.7\% & -14.8\% & -13.0\% & -11.9\% & -13.8\% & -13.6\% & -13.5\% \\
\bottomrule
\end{tabular}
}
\end{table}

\begin{table}[ht]
\centering
\caption{Geometric (GM) and arithmetic (AM) means over the 8 per-dataset performance ratios, expressed as a relative improvement, of EigenLI-32 on ViDoRe-v3.}
\label{tab:vidore_geomean_32}
\begin{tabular}{lrrrr}
\toprule
Method & R@10 & R@100 & nDCG@10 & nDCG@100 \\
\hline
GM vs KMeans++-32 & +4.8\% & +2.0\% & +5.4\% & +4.3\% \\
AM vs KMeans++-32 & +4.8\% & +2.0\% & +5.5\% & +4.4\% \\
\hline
GM vs Ward-32 & +3.1\% & +1.2\% & +3.4\% & +2.6\% \\
AM vs Ward-32 & +3.1\% & +1.2\% & +3.4\% & +2.6\% \\
\hline
GM vs Full MaxSim & -4.3\% & -1.7\% & -5.5\% & -4.7\% \\
AM vs Full MaxSim & -4.2\% & -1.7\% & -5.5\% & -4.7\% \\
\hline
\end{tabular}
\end{table}
\begin{table}[ht]
\centering
\caption{Geometric (GM) and arithmetic (AM) means over the 13
per-dataset relative improvements of EigenLI-SV-32 over MUVERA on
BEIR. For AnswerAI-ColBERT-small and GTE-ModernColBERT, we use centered MUVERA. The ColBERTv2 comparison
uses the original uncentered MUVERA representation.}
\label{tab:sv_geomean}
\small
\resizebox{\textwidth}{!}{%
\begin{tabular}{llrrrrrrrrr}
\toprule
Model & Mean & R@10 & R@100 & R@1000
      & nDCG@10 & nDCG@100 & nDCG@1000
      & MRR@10 & MRR@100 & MRR@1000 \\
\midrule
ColBERTv2
  & GM & +60.3\% & +43.9\% & +24.5\%
       & +69.6\% & +62.2\% & +51.4\%
       & +60.4\% & +57.0\% & +56.8\% \\
ColBERTv2
  & AM & +69.2\% & +62.7\% & +35.5\%
       & +78.9\% & +76.2\% & +61.3\%
       & +65.5\% & +61.4\% & +61.2\% \\
\midrule
AnswerAI-ColBERT-small
  & GM & +109.1\% & +71.7\% & +47.8\%
       & +125.6\% & +104.6\% & +89.3\%
       & +115.4\% & +106.0\% & +105.3\% \\
AnswerAI-ColBERT-small
  & AM & +153.7\% & +96.2\% & +67.8\%
       & +178.1\% & +138.6\% & +113.8\%
       & +161.1\% & +142.4\% & +141.0\% \\
\midrule
GTE-ModernColBERT
  & GM & +95.6\% & +72.5\% & +48.3\%
       & +104.6\% & +95.6\% & +84.6\%
       & +91.0\% & +87.3\% & +87.0\% \\
GTE-ModernColBERT
  & AM & +142.9\% & +97.5\% & +66.9\%
       & +156.0\% & +134.5\% & +118.1\%
       & +132.2\% & +126.0\% & +125.2\% \\
\bottomrule
\end{tabular}%
}
\end{table}
\paragraph{Comparison with MUVERA: } The performance difference is larger when compared to MUVERA. The numbers we obtain for MUVERA on MS MARCO are lower than those reported in the MUVERA paper, despite using its original source code. For MUVERA, we use 20 repetitions, a projection dimension of 16, and 32 partitions, giving a final FDE dimension of 10240. EigenLI-SV outperforms MUVERA across the evaluated BEIR datasets. The arithmetic-mean improvement in nDCG@10 is $78.9\%$ for ColBERTv2, $178.1\%$ for AnswerAI-ColBERT-small, and $156.0\%$ for GTE-ModernColBERT. For anisotropic models, these gains are noted against centered MUVERA (described in the following paragraphs). Without centering, MUVERA degrades by an additional order of magnitude; importantly, EigenLI-SV substantially outperforms even centered MUVERA. We refer the reader to \Cref{tab:sv_geomean} for the complete summary and discuss the interpretation of relative improvements alongside the absolute effectiveness values in \Cref{sec:discussion}.

\paragraph{Performance on anisotropic models and centering: }
Recent investigations have shown that MUVERA can degrade on modern models such as GTE-ModernColBERT~\cite{lateonregularization}, primarily due to the anisotropy of such models, wherein the token vectors are concentrated in a narrow cone. MUVERA depends on a SimHash projection whose approximation quality relies on random hyperplanes separating token vectors. This becomes less effective for highly anisotropic embeddings and may require a larger FDE dimension to retain quality. EigenLI-SV is less affected in our experiments, although the multi-vector EigenLI variant also exhibits model-dependent behavior.

We evaluate whether centering without renormalization mitigates the effect of anisotropy for AnswerAI-ColBERT-small and GTE-ModernColBERT. For the centered MUVERA baseline, we center each query and document
independently. Given the token vectors $\{x_i\}_{i=1}^{n}$ of a document or query, we compute their mean
$\bar{x}$ and replace each token vector with $x_i-\bar{x}$. Padding tokens are excluded when computing the
mean, and the centered vectors are not subsequently renormalized. We observe that centering substantially improves MUVERA for every model-level macro average across Recall, nDCG, and MRR at all three cutoffs. In contrast, EigenLI-32 is already robust without centering: across the same nine metric--cutoff combinations, its average barely changes on AnswerAI-ColBERT-small and decreases mildly on GTE-ModernColBERT. Centered MUVERA also remains below uncentered EigenLI-32 in nDCG@10. \Cref{tab:centering_macro} reports all average trends. In the Appendix, \Cref{tab:anisotropic_centering_at10} contains a detailed analysis of the effect of centering on anisotropic models.

\begin{table}[t]
\centering
\caption{Mean trends in MUVERA and EigenLI-SV (32) over 13 BEIR datasets before (U) and after (C) centering, with $\Delta=C-U$.}
\label{tab:centering_macro}
\scriptsize
\setlength{\tabcolsep}{3pt}
\resizebox{\textwidth}{!}{%
\begin{tabular}{lllrrr|rrr|rrr}
\toprule
Model & Metric & Method & \multicolumn{3}{c|}{@10} & \multicolumn{3}{c|}{@100} & \multicolumn{3}{c}{@1000} \\
\cmidrule(lr){4-6}\cmidrule(lr){7-9}\cmidrule(lr){10-12}
 & & & U & C & $\Delta$ & U & C & $\Delta$ & U & C & $\Delta$ \\
\midrule
\multirow{6}{*}{AnswerAI-ColBERT-small}
 & \multirow{2}{*}{Recall} & EigenLI-32 & 0.5725 & 0.5761 & +0.0036 & 0.7205 & 0.7278 & +0.0073 & 0.8376 & 0.8444 & +0.0068 \\
 & & MUVERA & 0.0302 & 0.3576 & +0.3274 & 0.1013 & 0.5053 & +0.4040 & 0.2571 & 0.6494 & +0.3923 \\
 & \multirow{2}{*}{nDCG} & EigenLI-32 & 0.4875 & 0.4871 & -0.0004 & 0.5197 & 0.5217 & +0.0020 & 0.5423 & 0.5438 & +0.0015 \\
 & & MUVERA & 0.0174 & 0.2866 & +0.2692 & 0.0339 & 0.3148 & +0.2809 & 0.0591 & 0.3382 & +0.2791 \\
 & \multirow{2}{*}{MRR} & EigenLI-32 & 0.5579 & 0.5572 & -0.0007 & 0.5643 & 0.5633 & -0.0010 & 0.5646 & 0.5636 & -0.0010 \\
 & & MUVERA & 0.0184 & 0.3272 & +0.3088 & 0.0231 & 0.3366 & +0.3135 & 0.0241 & 0.3371 & +0.3130 \\
\midrule
\multirow{6}{*}{GTE-ModernColBERT}
 & \multirow{2}{*}{Recall} & EigenLI-32 & 0.5898 & 0.5700 & -0.0198 & 0.7505 & 0.7292 & -0.0213 & 0.8615 & 0.8425 & -0.0190 \\
 & & MUVERA & 0.0247 & 0.3678 & +0.3431 & 0.0802 & 0.5009 & +0.4207 & 0.2384 & 0.6427 & +0.4043 \\
 & \multirow{2}{*}{nDCG} & EigenLI-32 & 0.4911 & 0.4753 & -0.0158 & 0.5317 & 0.5114 & -0.0203 & 0.5544 & 0.5326 & -0.0218 \\
 & & MUVERA & 0.0168 & 0.2929 & +0.2761 & 0.0281 & 0.3177 & +0.2896 & 0.0521 & 0.3417 & +0.2896 \\
 & \multirow{2}{*}{MRR} & EigenLI-32 & 0.5545 & 0.5407 & -0.0138 & 0.5611 & 0.5476 & -0.0135 & 0.5614 & 0.5479 & -0.0135 \\
 & & MUVERA & 0.0288 & 0.3416 & +0.3128 & 0.0326 & 0.3495 & +0.3169 & 0.0335 & 0.3500 & +0.3165 \\
\bottomrule
\end{tabular}
}
\end{table}

\paragraph{Compatibility with vector quantization: }
EigenLI-SV produces ordinary single vectors and can therefore use standard vector quantizers. We quantize EigenLI-SV-32 derived from ColBERT-v2 with FAISS \texttt{IndexPQ} at 1, 2, 4, and 8 bits per dimension. \Cref{tab:sv_quant_recall,tab:sv_quant_ndcg} report the arithmetic mean of the per-dataset relative Recall and nDCG losses against the corresponding full-MaxSim results for ColBERT v2, using the same evaluation queries. Results are noted for ArguAna, FiQA, SCIDOCS, and SciFact, but the same qualitative trends also hold in our completed ColBERTv2 runs on larger BEIR datasets.

\begin{table}[t]
\centering
\small
\setlength{\tabcolsep}{5pt}
\caption{Mean relative Recall loss (\%) against full MaxSim over the four noted datasets. Each entry reports @10/@100/@1000.}
\label{tab:sv_quant_recall}
\begin{tabular}{llrrrr}
\toprule
Model & Quantizer & 1 bit & 2 bits & 4 bits & 8 bits \\
\midrule
ColBERTv2 & PQ & -8.9/-7.8/-4.7 & -5.5/-6.1/-3.8 & -5.5/-5.6/-3.5 & -5.2/-5.7/-3.4 \\
\bottomrule
\end{tabular}
\end{table}

\begin{table}[t]
\centering
\small
\setlength{\tabcolsep}{5pt}
\caption{Mean relative nDCG loss (\%) against full MaxSim over the four noted datasets. Each entry reports @10/@100/@1000.}
\label{tab:sv_quant_ndcg}
\begin{tabular}{llrrrr}
\toprule
Model & Quantizer & 1 bit & 2 bits & 4 bits & 8 bits \\
\midrule
ColBERTv2 & PQ & -8.9/-8.4/-7.5 & -4.4/-4.7/-4.2 & -3.9/-4.1/-3.7 & -3.6/-3.9/-3.5 \\
\bottomrule
\end{tabular}
\end{table}

\paragraph{Compression cost: }
EigenLI also reduces offline compression cost relative to the clustering baselines. Across the 39 model--dataset pairs (three models and 13 BEIR datasets), $k$-means++ takes $6.81\times$ longer than EigenLI under the arithmetic mean of the per-pair time ratios ($6.47\times$ geometric mean), while Ward takes $17.46\times$ longer ($16.89\times$ geometric mean). \Cref{tab:compression_overhead} separates these ratios by model.

\begin{table}[ht]
\centering
\caption{Compression-time overhead relative to EigenLI. Each model row averages per-dataset time ratios over 13 BEIR datasets; the final row averages over all 39 model--dataset pairs. Ratios above one mean the clustering baseline is slower.}
\label{tab:compression_overhead}
\small
\begin{tabular}{lrrrr}
\toprule
Model & KMeans++ AM & KMeans++ GM & Ward AM & Ward GM \\
\midrule
ColBERTv2 & 6.34$\times$ & 5.97$\times$ & 16.91$\times$ & 16.41$\times$ \\
AnswerAI-ColBERT-small & 7.90$\times$ & 7.65$\times$ & 18.62$\times$ & 18.08$\times$ \\
GTE-ModernColBERT & 6.20$\times$ & 5.93$\times$ & 16.85$\times$ & 16.24$\times$ \\
\midrule
All model--dataset pairs & 6.81$\times$ & 6.47$\times$ & 17.46$\times$ & 16.89$\times$ \\
\bottomrule
\end{tabular}
\end{table}

\begin{table}[ht]
\centering
\caption{ColBERTv2 Recall across BEIR datasets.}
\label{tab:beir_recall}
\resizebox{\textwidth}{!}{%
\small
\begin{tabular}{l|l|ccc|c|ccc|c|ccc|c}
\toprule
\multirow{2}{*}{Dataset} & \multirow{2}{*}{Method} 
& \multicolumn{4}{c|}{Recall@10} 
& \multicolumn{4}{c|}{Recall@100} 
& \multicolumn{4}{c}{Recall@1000} \\
\cline{3-14}
 & 
& 8 & 16 & 32 & Full
& 8 & 16 & 32 & Full
& 8 & 16 & 32 & Full \\
\midrule

\multirow{3}{*}{ArguAna}
  & EigenLI & \textbf{0.6333} & \textbf{0.6500} & 0.6600 & \multirow{3}{*}{0.6667} & \textbf{0.9100} & 0.8967 & 0.9033 & \multirow{3}{*}{0.9167} & \textbf{0.9800} & 0.9800 & 0.9767 & \multirow{3}{*}{0.9833} \\
  & KMeans++ & 0.5833 & 0.6033 & 0.6567 & & 0.8533 & 0.8900 & 0.8967 & & 0.9733 & 0.9800 & \textbf{0.9867} & \\
  & Ward & 0.5800 & 0.6300 & \textbf{0.6633} & & 0.8733 & \textbf{0.9000} & \textbf{0.9200} & & 0.9767 & \textbf{0.9833} & \textbf{0.9867} & \\
\hline

\multirow{3}{*}{Climate-FEVER}
  & EigenLI & \textbf{0.2406} & \textbf{0.2499} & \textbf{0.2417} & \multirow{3}{*}{0.2782} & \textbf{0.4178} & \textbf{0.4465} & \textbf{0.4487} & \multirow{3}{*}{0.5077} & \textbf{0.6511} & \textbf{0.6639} & \textbf{0.6577} & \multirow{3}{*}{0.7123} \\
  & KMeans++ & 0.1965 & 0.2188 & 0.1981 & & 0.3611 & 0.3584 & 0.3676 & & 0.5670 & 0.5857 & 0.5990 & \\
  & Ward & 0.2094 & 0.2178 & 0.2188 & & 0.3657 & 0.3834 & 0.4111 & & 0.5756 & 0.5997 & 0.6127 & \\
\hline

\multirow{3}{*}{DBPedia-Entity}
  & EigenLI & \textbf{0.4688} & \textbf{0.5275} & \textbf{0.5364} & \multirow{3}{*}{0.5471} & \textbf{0.5309} & \textbf{0.6084} & \textbf{0.6236} & \multirow{3}{*}{0.6364} & \textbf{0.7445} & \textbf{0.7954} & 0.7912 & \multirow{3}{*}{0.8244} \\
  & KMeans++ & 0.4105 & 0.4760 & 0.5095 & & 0.4673 & 0.5470 & 0.6012 & & 0.6928 & 0.7650 & 0.7967 & \\
  & Ward & 0.4460 & 0.4851 & 0.5189 & & 0.5055 & 0.5674 & 0.6158 & & 0.7244 & 0.7729 & \textbf{0.8118} & \\
\hline

\multirow{3}{*}{FEVER}
  & EigenLI & \textbf{0.8099} & \textbf{0.8595} & \textbf{0.8764} & \multirow{3}{*}{0.9182} & \textbf{0.9158} & \textbf{0.9367} & \textbf{0.9445} & \multirow{3}{*}{0.9552} & \textbf{0.9483} & \textbf{0.9612} & \textbf{0.9667} & \multirow{3}{*}{0.9719} \\
  & KMeans++ & 0.6117 & 0.6884 & 0.7658 & & 0.8070 & 0.8635 & 0.9079 & & 0.9053 & 0.9310 & 0.9511 & \\
  & Ward & 0.7319 & 0.7894 & 0.8257 & & 0.8896 & 0.9115 & 0.9260 & & 0.9397 & 0.9498 & 0.9562 & \\
\hline

\multirow{3}{*}{FiQA}
 & EigenLI & \textbf{0.3230} & \textbf{0.3647} & \textbf{0.4125} & \multirow{3}{*}{0.4392} & \textbf{0.5471} & \textbf{0.6215} & \textbf{0.6363} & \multirow{3}{*}{0.6589} & \textbf{0.7937} & \textbf{0.8200} & \textbf{0.8229} & \multirow{3}{*}{0.8515} \\
 & KMeans++ & 0.2298 & 0.2848 & 0.3302 & & 0.4512 & 0.5006 & 0.5806 & & 0.7115 & 0.7619 & 0.8047 & \\
 & Ward & 0.2664 & 0.3175 & 0.3530 & & 0.4782 & 0.5453 & 0.6050 & & 0.7243 & 0.7744 & 0.8174 & \\
\hline

\multirow{3}{*}{HotpotQA}
  & EigenLI & \textbf{0.5777} & \textbf{0.6382} & \textbf{0.6525} & \multirow{3}{*}{0.6806} & \textbf{0.6977} & \textbf{0.7570} & \textbf{0.7732} & \multirow{3}{*}{0.7922} & \textbf{0.7988} & \textbf{0.8480} & \textbf{0.8589} & \multirow{3}{*}{0.8698} \\
  & KMeans++ & 0.4763 & 0.5529 & 0.6251 & & 0.6188 & 0.6879 & 0.7458 & & 0.7390 & 0.7940 & 0.8345 & \\
  & Ward & 0.5375 & 0.5988 & 0.6440 & & 0.6594 & 0.7162 & 0.7584 & & 0.7628 & 0.8102 & 0.8441 & \\
\hline

\multirow{3}{*}{MS MARCO}
  & EigenLI & \textbf{0.5907} & \textbf{0.6589} & 0.6685 & \multirow{3}{*}{0.6884} & \textbf{0.8562} & \textbf{0.8976} & 0.9021 & \multirow{3}{*}{0.9159} & \textbf{0.9624} & \textbf{0.9796} & 0.9802 & \multirow{3}{*}{0.9844} \\
  & KMeans++ & 0.4905 & 0.5916 & 0.6585 & & 0.7757 & 0.8551 & 0.9018 & & 0.9332 & 0.9648 & 0.9799 & \\
  & Ward & 0.5660 & 0.6434 & \textbf{0.6754} & & 0.8397 & 0.8893 & \textbf{0.9119} & & 0.9534 & 0.9758 & \textbf{0.9830} & \\
\hline

\multirow{3}{*}{NQ}
  & EigenLI & \textbf{0.6383} & \textbf{0.7006} & \textbf{0.7644} & \multirow{3}{*}{0.7872} & \textbf{0.8911} & \textbf{0.9239} & 0.9300 & \multirow{3}{*}{0.9528} & \textbf{0.9689} & \textbf{0.9889} & \textbf{0.9922} & \multirow{3}{*}{0.9939} \\
  & KMeans++ & 0.5128 & 0.5872 & 0.7033 & & 0.7867 & 0.8600 & 0.9056 & & 0.9094 & 0.9561 & 0.9828 & \\
  & Ward & 0.5694 & 0.6728 & 0.7250 & & 0.8072 & 0.8933 & \textbf{0.9339} & & 0.9317 & 0.9661 & 0.9794 & \\
\hline

\multirow{3}{*}{Quora}
  & EigenLI & \textbf{0.9276} & \textbf{0.9412} & 0.9206 & \multirow{3}{*}{0.9346} & \textbf{0.9872} & \textbf{0.9901} & 0.9806 & \multirow{3}{*}{0.9889} & \textbf{0.9986} & \textbf{0.9991} & 0.9965 & \multirow{3}{*}{0.9988} \\
  & KMeans++ & 0.9143 & 0.9347 & \textbf{0.9348} & & 0.9806 & 0.9882 & \textbf{0.9889} & & 0.9969 & 0.9987 & \textbf{0.9988} & \\
  & Ward & 0.9250 & 0.9348 & 0.9347 & & 0.9850 & 0.9886 & 0.9889 & & 0.9983 & 0.9987 & \textbf{0.9988} & \\
\hline

\multirow{3}{*}{SCIDOCS}
 & EigenLI & \textbf{0.1415} & \textbf{0.1370} & \textbf{0.1450} & \multirow{3}{*}{0.1535} & \textbf{0.3115} & \textbf{0.3215} & 0.3180 & \multirow{3}{*}{0.3345} & \textbf{0.5613} & \textbf{0.5703} & 0.5505 & \multirow{3}{*}{0.5653} \\
 & KMeans++ & 0.1035 & 0.1290 & 0.1315 & & 0.2895 & 0.3015 & 0.3115 & & 0.5323 & 0.5628 & 0.5713 & \\
 & Ward & 0.1230 & 0.1270 & 0.1410 & & 0.2995 & 0.3155 & \textbf{0.3295} & & 0.5588 & 0.5688 & \textbf{0.5793} & \\
\hline

\multirow{3}{*}{SciFact}
  & EigenLI & \textbf{0.7257} & \textbf{0.7684} & \textbf{0.7951} & \multirow{3}{*}{0.8046} & \textbf{0.8599} & \textbf{0.8860} & \textbf{0.8943} & \multirow{3}{*}{0.9287} & 0.9393 & 0.9527 & \textbf{0.9727} & \multirow{3}{*}{0.9733} \\
  & KMeans++ & 0.6782 & 0.7401 & 0.7583 & & 0.8476 & 0.8682 & 0.8799 & & 0.9317 & 0.9433 & 0.9567 & \\
  & Ward & 0.6901 & 0.7418 & 0.7659 & & 0.8342 & 0.8816 & 0.8910 & & \textbf{0.9567} & \textbf{0.9567} & 0.9533 & \\
\hline

\multirow{3}{*}{TREC-COVID}
  & EigenLI & \textbf{0.5580} & \textbf{0.7180} & \textbf{0.7960} & \multirow{3}{*}{0.7980} & \textbf{0.3224} & \textbf{0.4644} & \textbf{0.5300} & \multirow{3}{*}{0.5664} & \textbf{0.2211} & \textbf{0.3342} & \textbf{0.3943} & \multirow{3}{*}{0.4452} \\
  & KMeans++ & 0.3300 & 0.4240 & 0.5920 & & 0.1638 & 0.2438 & 0.3756 & & 0.1365 & 0.1837 & 0.2864 & \\
  & Ward & 0.4080 & 0.5740 & 0.6680 & & 0.2428 & 0.3418 & 0.4518 & & 0.1772 & 0.2583 & 0.3438 & \\
\hline

\multirow{3}{*}{Webis-Touche2020}
  & EigenLI & \textbf{0.2507} & \textbf{0.2417} & \textbf{0.2784} & \multirow{3}{*}{0.2603} & \textbf{0.4463} & \textbf{0.4772} & \textbf{0.4971} & \multirow{3}{*}{0.4922} & \textbf{0.7887} & \textbf{0.8223} & \textbf{0.8243} & \multirow{3}{*}{0.8275} \\
  & KMeans++ & 0.1350 & 0.1601 & 0.1867 & & 0.2816 & 0.3499 & 0.4011 & & 0.6212 & 0.6835 & 0.7239 & \\
  & Ward & 0.1771 & 0.1835 & 0.2245 & & 0.3535 & 0.3934 & 0.4175 & & 0.6748 & 0.7297 & 0.7618 & \\
\bottomrule
\end{tabular}
}
\end{table}

\begin{table}[ht]
\centering
\caption{Recall across ViDoRe-v3 datasets.}
\label{tab:vidore_recall}
\resizebox{\textwidth}{!}{%
\small
\setlength{\tabcolsep}{12pt} 
\begin{tabular}{l|l|ccc|c|ccc|c}
\toprule
\multirow{2}{*}{Dataset} & \multirow{2}{*}{Method} 
& \multicolumn{4}{c|}{Recall@10} 
& \multicolumn{4}{c}{Recall@100} \\
\cline{3-10}
 & 
& 8 & 16 & 32 & Full
& 8 & 16 & 32 & Full \\
\midrule
\multirow{3}{*}{HR}
  & EigenLI   & \textbf{0.5547} & \textbf{0.5871} & \textbf{0.6127} & \multirow{3}{*}{0.649} & \textbf{0.8641} & \textbf{0.8929} & \textbf{0.9083} & \multirow{3}{*}{0.92} \\
  & KMeans++  & 0.5134 & 0.5375 & 0.5831 & & 0.8568 & 0.8681 & 0.8905 & \\
  & Ward      & 0.5204 & 0.5536 & 0.5905 & & 0.8452 & 0.8721 & 0.8962 & \\
\hline
\multirow{3}{*}{Fin-EN}
  & EigenLI   & \textbf{0.5097} & \textbf{0.5810} & \textbf{0.6368} & \multirow{3}{*}{0.680} & \textbf{0.8187} & \textbf{0.8620} & \textbf{0.8980} & \multirow{3}{*}{0.93} \\
  & KMeans++  & 0.4667 & 0.5267 & 0.5942 & & 0.7883 & 0.8256 & 0.8726 & \\
  & Ward      & 0.4776 & 0.5470 & 0.6172 & & 0.7862 & 0.8286 & 0.8874 & \\
\hline
\multirow{3}{*}{Fin-FR}
  & EigenLI   & \textbf{0.3680} & \textbf{0.4518} & \textbf{0.5057} & \multirow{3}{*}{0.553} & \textbf{0.7633} & \textbf{0.8137} & \textbf{0.8490} & \multirow{3}{*}{0.87} \\
  & KMeans++  & 0.3326 & 0.4082 & 0.4616 & & 0.7140 & 0.7734 & 0.8053 & \\
  & Ward      & 0.3354 & 0.3991 & 0.4608 & & 0.7007 & 0.7736 & 0.8222 & \\
\hline
\multirow{3}{*}{Ind.}
  & EigenLI   & \textbf{0.4582} & \textbf{0.4958} & \textbf{0.5238} & \multirow{3}{*}{0.561} & \textbf{0.7519} & \textbf{0.7845} & \textbf{0.8031} & \multirow{3}{*}{0.83} \\
  & KMeans++  & 0.4340 & 0.4800 & 0.5003 & & 0.7328 & 0.7578 & 0.7763 & \\
  & Ward      & 0.4329 & 0.4760 & 0.5185 & & 0.7425 & 0.7616 & 0.7870 & \\
\hline
\multirow{3}{*}{Phar.}
  & EigenLI   & \textbf{0.6222} & \textbf{0.6522} & \textbf{0.6745} & \multirow{3}{*}{0.689} & 0.8814 & \textbf{0.9041} & \textbf{0.9142} & \multirow{3}{*}{0.93} \\
  & KMeans++  & 0.6020 & 0.6422 & 0.6543 & & \textbf{0.8746} & 0.8963 & 0.9095 & \\
  & Ward      & 0.5968 & 0.6373 & 0.6672 & & 0.8728 & 0.8933 & 0.9130 & \\
\hline
\multirow{3}{*}{CS}
  & EigenLI   & \textbf{0.7485} & \textbf{0.7764} & \textbf{0.7957} & \multirow{3}{*}{0.793} & \textbf{0.9689} & \textbf{0.9773} & \textbf{0.9823} & \multirow{3}{*}{0.99} \\
  & KMeans++  & 0.7278 & 0.7496 & 0.7598 & & 0.9625 & 0.9675 & 0.9767 & \\
  & Ward      & 0.7203 & 0.7473 & 0.7639 & & 0.9620 & 0.9727 & 0.9794 & \\
\hline
\multirow{3}{*}{Ener.}
  & EigenLI   & \textbf{0.6708} & \textbf{0.7052} & \textbf{0.7189} & \multirow{3}{*}{0.749} & \textbf{0.9284} & \textbf{0.9385} & \textbf{0.9466} & \multirow{3}{*}{0.95} \\
  & KMeans++  & 0.6205 & 0.6631 & 0.6904 & & 0.9039 & 0.9222 & 0.9460 & \\
  & Ward      & 0.6293 & 0.6814 & 0.7049 & & 0.9125 & 0.9279 & 0.9405 & \\
\hline
\multirow{3}{*}{Phys.}
  & EigenLI   & \textbf{0.4955} & \textbf{0.5041} & 0.5129 & \multirow{3}{*}{0.517} & \textbf{0.8824} & \textbf{0.8934} & \textbf{0.8975} & \multirow{3}{*}{0.90} \\
  & KMeans++  & 0.4770 & 0.4940 & 0.5114 & & 0.8693 & 0.8789 & 0.8874 & \\
  & Ward      & 0.4920 & 0.4962 & \textbf{0.5133} & & 0.8669 & 0.8776 & 0.8900 & \\
\bottomrule
\end{tabular}
}
\end{table}

\begin{table}[ht]
\centering
\caption{MUVERA (10240) vs EigenLI-SV (8256) Recall on all 13 BEIR datasets.}
\label{tab:sv_recall}
\small
\begin{tabular}{l|cc|cc|cc}
\toprule
\multirow{2}{*}{Dataset} & \multicolumn{2}{c|}{Recall@10} & \multicolumn{2}{c|}{Recall@100} & \multicolumn{2}{c}{Recall@1000} \\
\cline{2-7}
 & MUVERA & EigenLI-SV & MUVERA & EigenLI-SV & MUVERA & EigenLI-SV \\
\midrule
ArguAna & 0.5000 & \textbf{0.6600} & 0.8033 & \textbf{0.9033} & 0.9533 & \textbf{0.9767} \\
FiQA & 0.2399 & \textbf{0.4125} & 0.4359 & \textbf{0.6363} & 0.7033 & \textbf{0.8229} \\
SCIDOCS & 0.0935 & \textbf{0.1450} & 0.2775 & \textbf{0.3180} & 0.5485 & \textbf{0.5505} \\
SciFact & 0.6427 & \textbf{0.7951} & 0.8309 & \textbf{0.8943} & 0.9250 & \textbf{0.9727} \\
TREC-COVID & 0.2280 & \textbf{0.7960} & 0.0996 & \textbf{0.5300} & 0.0999 & \textbf{0.3943} \\
Webis-Touche2020 & 0.1042 & \textbf{0.2784} & 0.2338 & \textbf{0.4971} & 0.5508 & \textbf{0.8243} \\
Quora & 0.8603 & \textbf{0.9206} & 0.9552 & \textbf{0.9806} & 0.9905 & \textbf{0.9965} \\
NQ & 0.4600 & \textbf{0.7644} & 0.7994 & \textbf{0.9300} & 0.9306 & \textbf{0.9922} \\
HotpotQA & 0.4271 & \textbf{0.6525} & 0.5672 & \textbf{0.7732} & 0.6989 & \textbf{0.8589} \\
DBPedia-Entity & 0.3630 & \textbf{0.5364} & 0.4261 & \textbf{0.6236} & 0.6518 & \textbf{0.7912} \\
Climate-FEVER & 0.1554 & \textbf{0.2417} & 0.2803 & \textbf{0.4487} & 0.5042 & \textbf{0.6577} \\
FEVER & 0.5981 & \textbf{0.8764} & 0.8070 & \textbf{0.9445} & 0.9106 & \textbf{0.9667} \\
MS MARCO & 0.5377 & \textbf{0.6685} & 0.8141 & \textbf{0.9021} & 0.9373 & \textbf{0.9802} \\
\bottomrule
\end{tabular}
\end{table}
\begin{table}[ht]
\centering
\caption{MUVERA (10240) vs EigenLI-SV (8256) nDCG on all 13 BEIR datasets.}
\label{tab:sv_ndcg}
\small
\begin{tabular}{l|cc|cc|cc}
\toprule
\multirow{2}{*}{Dataset} & \multicolumn{2}{c|}{nDCG@10} & \multicolumn{2}{c|}{nDCG@100} & \multicolumn{2}{c}{nDCG@1000} \\
\cline{2-7}
 & MUVERA & EigenLI-SV & MUVERA & EigenLI-SV & MUVERA & EigenLI-SV \\
\midrule
ArguAna & 0.2833 & \textbf{0.4378} & 0.3489 & \textbf{0.4880} & 0.3676 & \textbf{0.4976} \\
FiQA & 0.1940 & \textbf{0.3425} & 0.2438 & \textbf{0.4019} & 0.2864 & \textbf{0.4318} \\
SCIDOCS & 0.0880 & \textbf{0.1493} & 0.1519 & \textbf{0.2080} & 0.2070 & \textbf{0.2558} \\
SciFact & 0.5156 & \textbf{0.6841} & 0.5569 & \textbf{0.7062} & 0.5689 & \textbf{0.7157} \\
TREC-COVID & 0.2254 & \textbf{0.7417} & 0.1120 & \textbf{0.5176} & 0.1063 & \textbf{0.4264} \\
Webis-Touche2020 & 0.0886 & \textbf{0.2866} & 0.1509 & \textbf{0.4011} & 0.2531 & \textbf{0.5090} \\
Quora & 0.7632 & \textbf{0.8360} & 0.7878 & \textbf{0.8528} & 0.7933 & \textbf{0.8557} \\
NQ & 0.3160 & \textbf{0.5466} & 0.3922 & \textbf{0.5832} & 0.4097 & \textbf{0.5919} \\
HotpotQA & 0.4104 & \textbf{0.6442} & 0.4455 & \textbf{0.6751} & 0.4652 & \textbf{0.6881} \\
DBPedia-Entity & 0.3884 & \textbf{0.5439} & 0.4014 & \textbf{0.5776} & 0.4709 & \textbf{0.6342} \\
Climate-FEVER & 0.1225 & \textbf{0.1970} & 0.1568 & \textbf{0.2519} & 0.1932 & \textbf{0.2884} \\
FEVER & 0.4324 & \textbf{0.7479} & 0.4772 & \textbf{0.7641} & 0.4909 & \textbf{0.7674} \\
MS MARCO & 0.3531 & \textbf{0.4464} & 0.4119 & \textbf{0.4976} & 0.4279 & \textbf{0.5079} \\
\bottomrule
\end{tabular}
\end{table}
\section{Discussion, Limitations, and Future Directions}\label{sec:discussion}

We introduce EigenLI, a new method for reducing the number of vectors used to represent each document in multi-vector retrieval. EigenLI is motivated by the empirical observation that token-level embeddings produced by multi-vector models often exhibit substantial low-rank structure. Rather than compressing a document into a smaller set of representative token vectors, EigenLI represents it using a low-dimensional spectral subspace and introduces a corresponding scoring function for measuring the interaction between query embeddings and this compressed representation.

Our experiments show that the effectiveness of this approach depends on the geometry of the underlying model. At $k = 32$, for ColBERTv2 and ColQwen3, EigenLI improves over strong training-free compression baselines based on Ward hierarchical pooling. On AnswerAI-ColBERT-small, EigenLI performs comparably to Ward, whereas on GTE-ModernColBERT it underperforms Ward. An additional benefit of the EigenLI scoring function is that it admits an exact high-dimensional single-vector dot-product representation, EigenLI-SV. At comparable representation dimensionality, EigenLI-SV substantially outperforms MUVERA in our experiments.

\paragraph{Limitations.}
Our experiments evaluate compression primarily using end-to-end brute-force retrieval. A practical large-scale EigenLI retrieval pipeline would instead use EigenLI-SV, potentially with quantization, together with an approximate nearest-neighbor (ANN) index to retrieve candidates, followed by reranking with full MaxSim. A controlled comparison of such a pipeline against optimized systems based on Ward pooling and PLAID remains an important missing experiment.

A second limitation is the dimensionality of EigenLI-SV. Its dimension scales quadratically with the token-vector dimension, which is manageable for models such as ColBERTv2 but becomes substantially more expensive for models with larger token embeddings. Combining EigenLI-SV with sketching, dimensionality reduction, or learned projections is therefore an important direction for improving its practical scalability. Finally, our current quantization experiments are limited in scope. A more systematic study is needed to understand how aggressively EigenLI/EigenLI-SV can be quantized, and how its quantization--effectiveness tradeoff compares with existing single- and multi-vector retrieval systems.

\paragraph{Future directions.}
Our results suggest several directions for future work.

\begin{enumerate}
    \item \textbf{Understanding model-dependent geometry.}
    Different multi-vector models appear to induce substantially different geometries in their token representations. What geometric properties determine whether spectral compression such as EigenLI or clustering-based compression such as Ward is preferable? Can the two approaches be combined to obtain a more robust compression method? A related question is whether EigenLI can be modified so that retrieval effectiveness is monotonic, as the number of retained directions increases.

    \item \textbf{Understanding the behavior of EigenLI on anisotropic models.}
    EigenLI underperforms Ward on GTE-ModernColBERT. Moreover, on some datasets, such as DBPedia-Entity, its performance can decrease substantially as $k$ increases from 8 to 32. This indicates that GTE-ModernColBERT exhibits a highly degenerate low rank structure. In our preliminary experiments, LateOn exhibited an even more pronounced low-rank behavior. Understanding these phenomena, and developing strategies that make spectral compression more robust to them, is an interesting direction for future work.

    \item \textbf{End-to-end ANN retrieval.}
    How does an EigenLI-SV-based ANN retrieval system, followed by full-MaxSim reranking, compare with highly optimized late-interaction retrieval systems such as PLAID and its subsequent implementations? Such a comparison should jointly consider retrieval effectiveness, latency, index size, construction cost, and memory usage.

    \item \textbf{Beyond low-rank structure.}
    Low-rank structure is only one possible structure in multi-vector representations. What other mathematical structure is present in token-level embedding sets, and can such structure be exploited to design more effective compression and indexing schemes?

    \item \textbf{Relationship to MUVERA and Chamfer approximation.}
    MUVERA has been proven to be optimal for approximating Chamfer scores \cite{j26, jayaram2026near}, yet EigenLI-SV performs substantially better empirically in our experiments. Understanding this gap is an interesting theoretical and empirical question. One possibility is that retrieval does not require uniformly accurate approximation of all query-document scores: it may suffice to preserve the relative ordering of the small set of documents relevant to each query. Another possibility is that EigenLI-SV gains from exploiting the low-rank structure present in actual multi-vector representations, whereas worst-case approximation guarantees do not assume such structure. This raises the question of whether MUVERA can be adapted to exploit low-rank structure and thereby obtain better practical efficiency and retrieval effectiveness.
\end{enumerate}

\newpage
\clearpage
\bibliography{references}
\bibliographystyle{neurips2026}

\newpage
\appendix
\section{Detailed ViDoRe-v3 Results}
\begin{table}[H]
\centering
\caption{nDCG across ViDoRe-v3 datasets.}
\label{tab:vidore_ndcg}
\resizebox{\textwidth}{!}{%
\small
\setlength{\tabcolsep}{12pt}
\begin{tabular}{l|l|ccc|c|ccc|c}
\toprule
\multirow{2}{*}{Dataset} & \multirow{2}{*}{Method}
& \multicolumn{4}{c|}{nDCG@10}
& \multicolumn{4}{c}{nDCG@100} \\
\cline{3-10}
 &
& 8 & 16 & 32 & Full
& 8 & 16 & 32 & Full \\
\midrule
\multirow{3}{*}{HR}
  & EigenLI   & \textbf{0.5009} & \textbf{0.5450} & \textbf{0.5672} & \multirow{3}{*}{0.602} & \textbf{0.5926} & \textbf{0.6360} & \textbf{0.6571} & \multirow{3}{*}{0.69} \\
  & KMeans++  & 0.4533 & 0.4901 & 0.5303 & & 0.5547 & 0.5899 & 0.6234 & \\
  & Ward      & 0.4611 & 0.5065 & 0.5440 & & 0.5586 & 0.6045 & 0.6384 & \\
\hline
\multirow{3}{*}{Fin-EN}
  & EigenLI   & \textbf{0.4605} & \textbf{0.5298} & \textbf{0.5823} & \multirow{3}{*}{0.636} & \textbf{0.5540} & \textbf{0.6164} & \textbf{0.6660} & \multirow{3}{*}{0.72} \\
  & KMeans++  & 0.4084 & 0.4755 & 0.5467 & & 0.5057 & 0.5662 & 0.6313 & \\
  & Ward      & 0.4167 & 0.4953 & 0.5737 & & 0.5117 & 0.5815 & 0.6589 & \\
\hline
\multirow{3}{*}{Fin-FR}
  & EigenLI   & \textbf{0.3013} & \textbf{0.3760} & \textbf{0.4200} & \multirow{3}{*}{0.468} & \textbf{0.4152} & \textbf{0.4842} & \textbf{0.5233} & \multirow{3}{*}{0.57} \\
  & KMeans++  & 0.2657 & 0.3316 & 0.3912 & & 0.3761 & 0.4382 & 0.4927 & \\
  & Ward      & 0.2717 & 0.3349 & 0.3895 & & 0.3784 & 0.4435 & 0.4962 & \\
\hline
\multirow{3}{*}{Ind.}
  & EigenLI   & \textbf{0.4118} & \textbf{0.4587} & \textbf{0.4896} & \multirow{3}{*}{0.537} & \textbf{0.5029} & \textbf{0.5474} & \textbf{0.5763} & \multirow{3}{*}{0.62} \\
  & KMeans++  & 0.3843 & 0.4268 & 0.4545 & & 0.4762 & 0.5153 & 0.5425 & \\
  & Ward      & 0.3747 & 0.4234 & 0.4722 & & 0.4690 & 0.5136 & 0.5579 & \\
\hline
\multirow{3}{*}{Phar.}
  & EigenLI   & \textbf{0.5678} & \textbf{0.6171} & \textbf{0.6365} & \multirow{3}{*}{0.658} & \textbf{0.6506} & \textbf{0.6994} & \textbf{0.7155} & \multirow{3}{*}{0.74} \\
  & KMeans++  & 0.5391 & 0.5841 & 0.6110 & & 0.6270 & 0.6672 & 0.6959 & \\
  & Ward      & 0.5464 & 0.5876 & 0.6294 & & 0.6344 & 0.6722 & 0.7117 & \\
\hline
\multirow{3}{*}{CS}
  & EigenLI   & \textbf{0.7156} & \textbf{0.7487} & \textbf{0.7631} & \multirow{3}{*}{0.756} & \textbf{0.7904} & \textbf{0.8177} & \textbf{0.8286} & \multirow{3}{*}{0.83} \\
  & KMeans++  & 0.6825 & 0.7076 & 0.7250 & & 0.7631 & 0.7832 & 0.7996 & \\
  & Ward      & 0.6743 & 0.7081 & 0.7197 & & 0.7565 & 0.7853 & 0.7944 & \\
\hline
\multirow{3}{*}{Ener.}
  & EigenLI   & \textbf{0.5500} & \textbf{0.6019} & \textbf{0.6221} & \multirow{3}{*}{0.667} & \textbf{0.6239} & \textbf{0.6705} & \textbf{0.6901} & \multirow{3}{*}{0.73} \\
  & KMeans++  & 0.5114 & 0.5541 & 0.5877 & & 0.5911 & 0.6291 & 0.6616 & \\
  & Ward      & 0.5120 & 0.5737 & 0.6047 & & 0.5925 & 0.6452 & 0.6736 & \\
\hline
\multirow{3}{*}{Phys.}
  & EigenLI   & \textbf{0.4758} & \textbf{0.4788} & 0.4848 & \multirow{3}{*}{0.492} & \textbf{0.6055} & \textbf{0.6110} & \textbf{0.6146} & \multirow{3}{*}{0.62} \\
  & KMeans++  & 0.4508 & 0.4721 & \textbf{0.4858} & & 0.5810 & 0.5995 & 0.6098 & \\
  & Ward      & 0.4694 & 0.4729 & 0.4840 & & 0.5929 & 0.5996 & 0.6087 & \\
\bottomrule
\end{tabular}
}
\end{table}

\clearpage
\section{Detailed ColBERTv2 BEIR Results}

\Cref{tab:beir_ndcg,tab:beir_mrr} complement the Recall results in
\Cref{tab:beir_recall}. They report the same evaluation protocol, compression
budgets, and uncompressed MaxSim reference at every cutoff.

\begin{landscape}
\begin{table}[p]
\centering
\caption{ColBERTv2 nDCG across BEIR datasets.}
\label{tab:beir_ndcg}
\resizebox{\linewidth}{!}{%
\renewcommand{\arraystretch}{0.92}
\small
\begin{tabular}{l|l|ccc|c|ccc|c|ccc|c}
\toprule
\multirow{2}{*}{Dataset} & \multirow{2}{*}{Method}
& \multicolumn{4}{c|}{nDCG@10}
& \multicolumn{4}{c|}{nDCG@100}
& \multicolumn{4}{c}{nDCG@1000} \\
\cline{3-14}
&
& 8 & 16 & 32 & Full
& 8 & 16 & 32 & Full
& 8 & 16 & 32 & Full \\
\midrule
\multirow{3}{*}{ArguAna}
  & EigenLI & \textbf{0.4092} & \textbf{0.4243} & 0.4378 & \multirow{3}{*}{0.4337} & \textbf{0.4681} & \textbf{0.4771} & 0.4880 & \multirow{3}{*}{0.4864} & \textbf{0.4767} & \textbf{0.4878} & 0.4976 & \multirow{3}{*}{0.4947} \\
  & KMeans++ & 0.3646 & 0.4082 & \textbf{0.4392} & & 0.4228 & 0.4694 & \textbf{0.4896} & & 0.4385 & 0.4804 & \textbf{0.5009} & \\
  & Ward & 0.3624 & 0.4145 & 0.4345 & & 0.4268 & 0.4721 & 0.4893 & & 0.4402 & 0.4825 & 0.4976 & \\
\hline
\multirow{3}{*}{Climate-FEVER}
  & EigenLI & \textbf{0.2140} & \textbf{0.2108} & \textbf{0.1970} & \multirow{3}{*}{0.2224} & \textbf{0.2617} & \textbf{0.2639} & \textbf{0.2519} & \multirow{3}{*}{0.2843} & \textbf{0.3019} & \textbf{0.3019} & \textbf{0.2884} & \multirow{3}{*}{0.3195} \\
  & KMeans++ & 0.1724 & 0.1850 & 0.1618 & & 0.2162 & 0.2223 & 0.2085 & & 0.2502 & 0.2603 & 0.2466 & \\
  & Ward & 0.1774 & 0.1802 & 0.1821 & & 0.2193 & 0.2249 & 0.2330 & & 0.2545 & 0.2609 & 0.2668 & \\
\hline
\multirow{3}{*}{DBPedia-Entity}
  & EigenLI & \textbf{0.4904} & \textbf{0.5280} & \textbf{0.5439} & \multirow{3}{*}{0.5553} & \textbf{0.5104} & \textbf{0.5610} & \textbf{0.5776} & \multirow{3}{*}{0.5903} & \textbf{0.5794} & \textbf{0.6220} & \textbf{0.6342} & \multirow{3}{*}{0.6502} \\
  & KMeans++ & 0.4126 & 0.4918 & 0.5235 & & 0.4312 & 0.5132 & 0.5639 & & 0.5029 & 0.5851 & 0.6266 & \\
  & Ward & 0.4661 & 0.4977 & 0.5296 & & 0.4862 & 0.5310 & 0.5659 & & 0.5537 & 0.5969 & 0.6297 & \\
\hline
\multirow{3}{*}{FEVER}
  & EigenLI & \textbf{0.6500} & \textbf{0.7196} & \textbf{0.7479} & \multirow{3}{*}{0.8371} & \textbf{0.6742} & \textbf{0.7376} & \textbf{0.7641} & \multirow{3}{*}{0.8464} & \textbf{0.6788} & \textbf{0.7411} & \textbf{0.7674} & \multirow{3}{*}{0.8490} \\
  & KMeans++ & 0.4550 & 0.5312 & 0.6206 & & 0.4974 & 0.5699 & 0.6527 & & 0.5103 & 0.5790 & 0.6587 & \\
  & Ward & 0.5795 & 0.6486 & 0.6913 & & 0.6147 & 0.6762 & 0.7143 & & 0.6214 & 0.6815 & 0.7186 & \\
\hline
\multirow{3}{*}{FiQA}
  & EigenLI & \textbf{0.2670} & \textbf{0.3142} & \textbf{0.3425} & \multirow{3}{*}{0.3627} & \textbf{0.3255} & \textbf{0.3812} & \textbf{0.4019} & \multirow{3}{*}{0.4228} & \textbf{0.3645} & \textbf{0.4131} & \textbf{0.4318} & \multirow{3}{*}{0.4533} \\
  & KMeans++ & 0.1848 & 0.2366 & 0.2840 & & 0.2418 & 0.2927 & 0.3492 & & 0.2825 & 0.3330 & 0.3841 & \\
  & Ward & 0.2202 & 0.2639 & 0.3078 & & 0.2751 & 0.3242 & 0.3746 & & 0.3141 & 0.3608 & 0.4079 & \\
\hline
\multirow{3}{*}{HotpotQA}
  & EigenLI & \textbf{0.5697} & \textbf{0.6298} & \textbf{0.6442} & \multirow{3}{*}{0.6793} & \textbf{0.6005} & \textbf{0.6601} & \textbf{0.6751} & \multirow{3}{*}{0.7079} & \textbf{0.6157} & \textbf{0.6739} & \textbf{0.6881} & \multirow{3}{*}{0.7197} \\
  & KMeans++ & 0.4558 & 0.5373 & 0.6172 & & 0.4918 & 0.5716 & 0.6480 & & 0.5100 & 0.5875 & 0.6614 & \\
  & Ward & 0.5274 & 0.5910 & 0.6376 & & 0.5583 & 0.6208 & 0.6670 & & 0.5738 & 0.6349 & 0.6799 & \\
\hline
\multirow{3}{*}{MS MARCO}
  & EigenLI & \textbf{0.3913} & \textbf{0.4403} & 0.4464 & \multirow{3}{*}{0.4656} & \textbf{0.4486} & \textbf{0.4921} & 0.4976 & \multirow{3}{*}{0.5154} & \textbf{0.4626} & \textbf{0.5029} & 0.5079 & \multirow{3}{*}{0.5244} \\
  & KMeans++ & 0.3141 & 0.3879 & 0.4406 & & 0.3747 & 0.4444 & 0.4939 & & 0.3949 & 0.4587 & 0.5042 & \\
  & Ward & 0.3720 & 0.4279 & \textbf{0.4539} & & 0.4309 & 0.4810 & \textbf{0.5054} & & 0.4456 & 0.4923 & \textbf{0.5148} & \\
\hline
\multirow{3}{*}{NQ}
  & EigenLI & \textbf{0.4348} & \textbf{0.5033} & \textbf{0.5466} & \multirow{3}{*}{0.5898} & \textbf{0.4914} & \textbf{0.5546} & \textbf{0.5832} & \multirow{3}{*}{0.6279} & \textbf{0.5017} & \textbf{0.5636} & \textbf{0.5919} & \multirow{3}{*}{0.6338} \\
  & KMeans++ & 0.3330 & 0.3863 & 0.4928 & & 0.3919 & 0.4481 & 0.5393 & & 0.4082 & 0.4608 & 0.5493 & \\
  & Ward & 0.3753 & 0.4718 & 0.5154 & & 0.4302 & 0.5222 & 0.5644 & & 0.4466 & 0.5320 & 0.5708 & \\
\hline
\multirow{3}{*}{Quora}
  & EigenLI & \textbf{0.8462} & \textbf{0.8619} & 0.8360 & \multirow{3}{*}{0.8525} & \textbf{0.8630} & \textbf{0.8761} & 0.8528 & \multirow{3}{*}{0.8680} & \textbf{0.8651} & \textbf{0.8779} & 0.8556 & \multirow{3}{*}{0.8699} \\
  & KMeans++ & 0.8264 & 0.8523 & \textbf{0.8526} & & 0.8446 & 0.8676 & \textbf{0.8681} & & 0.8474 & 0.8696 & \textbf{0.8700} & \\
  & Ward & 0.8413 & 0.8529 & 0.8526 & & 0.8580 & 0.8683 & 0.8681 & & 0.8604 & 0.8702 & 0.8700 & \\
\hline
\multirow{3}{*}{SCIDOCS}
  & EigenLI & \textbf{0.1268} & \textbf{0.1375} & \textbf{0.1493} & \multirow{3}{*}{0.1517} & \textbf{0.1842} & \textbf{0.2006} & \textbf{0.2080} & \multirow{3}{*}{0.2119} & \textbf{0.2352} & \textbf{0.2522} & \textbf{0.2558} & \multirow{3}{*}{0.2601} \\
  & KMeans++ & 0.0973 & 0.1178 & 0.1374 & & 0.1587 & 0.1747 & 0.1980 & & 0.2078 & 0.2278 & 0.2511 & \\
  & Ward & 0.1148 & 0.1317 & 0.1381 & & 0.1752 & 0.1961 & 0.2009 & & 0.2277 & 0.2477 & 0.2517 & \\
\hline
\multirow{3}{*}{SciFact}
  & EigenLI & \textbf{0.5918} & \textbf{0.6431} & \textbf{0.6841} & \multirow{3}{*}{0.6959} & \textbf{0.6223} & \textbf{0.6692} & \textbf{0.7062} & \multirow{3}{*}{0.7220} & \textbf{0.6323} & \textbf{0.6776} & \textbf{0.7157} & \multirow{3}{*}{0.7276} \\
  & KMeans++ & 0.5415 & 0.5926 & 0.6311 & & 0.5797 & 0.6223 & 0.6593 & & 0.5907 & 0.6318 & 0.6689 & \\
  & Ward & 0.5484 & 0.6071 & 0.6387 & & 0.5818 & 0.6382 & 0.6674 & & 0.5969 & 0.6470 & 0.6749 & \\
\hline
\multirow{3}{*}{TREC-COVID}
  & EigenLI & \textbf{0.5317} & \textbf{0.6633} & \textbf{0.7417} & \multirow{3}{*}{0.7369} & \textbf{0.3278} & \textbf{0.4566} & \textbf{0.5176} & \multirow{3}{*}{0.5410} & \textbf{0.2476} & \textbf{0.3661} & \textbf{0.4264} & \multirow{3}{*}{0.4713} \\
  & KMeans++ & 0.3283 & 0.4172 & 0.5560 & & 0.1777 & 0.2520 & 0.3708 & & 0.1490 & 0.2033 & 0.3113 & \\
  & Ward & 0.3931 & 0.5371 & 0.6205 & & 0.2482 & 0.3436 & 0.4415 & & 0.1944 & 0.2835 & 0.3715 & \\
\hline
\multirow{3}{*}{Webis-Touche2020}
  & EigenLI & \textbf{0.2406} & \textbf{0.2441} & \textbf{0.2866} & \multirow{3}{*}{0.2649} & \textbf{0.3440} & \textbf{0.3698} & \textbf{0.4011} & \multirow{3}{*}{0.3851} & \textbf{0.4575} & \textbf{0.4838} & \textbf{0.5090} & \multirow{3}{*}{0.4951} \\
  & KMeans++ & 0.1295 & 0.1554 & 0.1760 & & 0.1996 & 0.2485 & 0.2867 & & 0.3120 & 0.3596 & 0.3939 & \\
  & Ward & 0.1651 & 0.1788 & 0.2135 & & 0.2539 & 0.2878 & 0.3117 & & 0.3615 & 0.4002 & 0.4252 & \\
\bottomrule
\end{tabular}
}
\end{table}
\end{landscape}

\newpage
\begin{landscape}
\begin{table}[p]
\centering
\caption{ColBERTv2 MRR across BEIR datasets.}
\label{tab:beir_mrr}
\resizebox{\linewidth}{!}{%
\renewcommand{\arraystretch}{0.92}
\small
\begin{tabular}{l|l|ccc|c|ccc|c|ccc|c}
\toprule
\multirow{2}{*}{Dataset} & \multirow{2}{*}{Method}
& \multicolumn{4}{c|}{MRR@10}
& \multicolumn{4}{c|}{MRR@100}
& \multicolumn{4}{c}{MRR@1000} \\
\cline{3-14}
&
& 8 & 16 & 32 & Full
& 8 & 16 & 32 & Full
& 8 & 16 & 32 & Full \\
\midrule
\multirow{3}{*}{ArguAna}
  & EigenLI & \textbf{0.3399} & \textbf{0.3537} & 0.3680 & \multirow{3}{*}{0.3608} & \textbf{0.3518} & \textbf{0.3645} & 0.3777 & \multirow{3}{*}{0.3713} & \textbf{0.3521} & \textbf{0.3649} & 0.3781 & \multirow{3}{*}{0.3716} \\
  & KMeans++ & 0.2972 & 0.3476 & \textbf{0.3710} & & 0.3093 & 0.3602 & \textbf{0.3809} & & 0.3100 & 0.3605 & \textbf{0.3813} & \\
  & Ward & 0.2940 & 0.3471 & 0.3624 & & 0.3079 & 0.3588 & 0.3736 & & 0.3085 & 0.3592 & 0.3739 & \\
\hline
\multirow{3}{*}{Climate-FEVER}
  & EigenLI & \textbf{0.2821} & \textbf{0.2679} & \textbf{0.2428} & \multirow{3}{*}{0.2743} & \textbf{0.2900} & \textbf{0.2772} & \textbf{0.2519} & \multirow{3}{*}{0.2854} & \textbf{0.2912} & \textbf{0.2782} & \textbf{0.2529} & \multirow{3}{*}{0.2858} \\
  & KMeans++ & 0.2250 & 0.2325 & 0.1971 & & 0.2338 & 0.2396 & 0.2087 & & 0.2348 & 0.2408 & 0.2095 & \\
  & Ward & 0.2241 & 0.2222 & 0.2288 & & 0.2333 & 0.2310 & 0.2389 & & 0.2343 & 0.2319 & 0.2397 & \\
\hline
\multirow{3}{*}{DBPedia-Entity}
  & EigenLI & \textbf{0.8173} & 0.8135 & \textbf{0.8386} & \multirow{3}{*}{0.8662} & \textbf{0.8193} & 0.8135 & \textbf{0.8386} & \multirow{3}{*}{0.8662} & \textbf{0.8194} & 0.8135 & \textbf{0.8386} & \multirow{3}{*}{0.8662} \\
  & KMeans++ & 0.6989 & 0.8022 & 0.8269 & & 0.7016 & 0.8035 & 0.8269 & & 0.7016 & 0.8035 & 0.8269 & \\
  & Ward & 0.7858 & \textbf{0.8148} & 0.8259 & & 0.7867 & \textbf{0.8162} & 0.8259 & & 0.7867 & \textbf{0.8162} & 0.8259 & \\
\hline
\multirow{3}{*}{FEVER}
  & EigenLI & \textbf{0.6256} & \textbf{0.7036} & \textbf{0.7343} & \multirow{3}{*}{0.8423} & \textbf{0.6307} & \textbf{0.7070} & \textbf{0.7372} & \multirow{3}{*}{0.8435} & \textbf{0.6308} & \textbf{0.7071} & \textbf{0.7373} & \multirow{3}{*}{0.8435} \\
  & KMeans++ & 0.4271 & 0.5036 & 0.6000 & & 0.4356 & 0.5114 & 0.6066 & & 0.4360 & 0.5117 & 0.6067 & \\
  & Ward & 0.5567 & 0.6315 & 0.6763 & & 0.5640 & 0.6370 & 0.6807 & & 0.5642 & 0.6371 & 0.6808 & \\
\hline
\multirow{3}{*}{FiQA}
  & EigenLI & \textbf{0.3337} & \textbf{0.3928} & \textbf{0.4159} & \multirow{3}{*}{0.4315} & \textbf{0.3425} & \textbf{0.4029} & \textbf{0.4246} & \multirow{3}{*}{0.4391} & \textbf{0.3433} & \textbf{0.4035} & \textbf{0.4251} & \multirow{3}{*}{0.4396} \\
  & KMeans++ & 0.2345 & 0.2967 & 0.3481 & & 0.2449 & 0.3054 & 0.3577 & & 0.2456 & 0.3063 & 0.3583 & \\
  & Ward & 0.2823 & 0.3294 & 0.3758 & & 0.2911 & 0.3382 & 0.3857 & & 0.2919 & 0.3389 & 0.3863 & \\
\hline
\multirow{3}{*}{HotpotQA}
  & EigenLI & \textbf{0.7607} & \textbf{0.8275} & \textbf{0.8464} & \multirow{3}{*}{0.8836} & \textbf{0.7638} & \textbf{0.8294} & \textbf{0.8484} & \multirow{3}{*}{0.8850} & \textbf{0.7639} & \textbf{0.8295} & \textbf{0.8484} & \multirow{3}{*}{0.8851} \\
  & KMeans++ & 0.6224 & 0.7207 & 0.8154 & & 0.6274 & 0.7243 & 0.8177 & & 0.6276 & 0.7245 & 0.8178 & \\
  & Ward & 0.7152 & 0.7862 & 0.8371 & & 0.7186 & 0.7888 & 0.8394 & & 0.7188 & 0.7889 & 0.8395 & \\
\hline
\multirow{3}{*}{MS MARCO}
  & EigenLI & \textbf{0.3325} & \textbf{0.3751} & 0.3803 & \multirow{3}{*}{0.3994} & \textbf{0.3438} & \textbf{0.3853} & 0.3906 & \multirow{3}{*}{0.4093} & \textbf{0.3443} & \textbf{0.3857} & 0.3910 & \multirow{3}{*}{0.4096} \\
  & KMeans++ & 0.2621 & 0.3280 & 0.3760 & & 0.2738 & 0.3390 & 0.3867 & & 0.2745 & 0.3395 & 0.3871 & \\
  & Ward & 0.3152 & 0.3642 & \textbf{0.3880} & & 0.3268 & 0.3746 & \textbf{0.3982} & & 0.3273 & 0.3750 & \textbf{0.3985} & \\
\hline
\multirow{3}{*}{NQ}
  & EigenLI & \textbf{0.3916} & \textbf{0.4570} & \textbf{0.4973} & \multirow{3}{*}{0.5446} & \textbf{0.4016} & \textbf{0.4667} & \textbf{0.5030} & \multirow{3}{*}{0.5522} & \textbf{0.4019} & \textbf{0.4670} & \textbf{0.5033} & \multirow{3}{*}{0.5524} \\
  & KMeans++ & 0.2941 & 0.3416 & 0.4425 & & 0.3046 & 0.3528 & 0.4509 & & 0.3051 & 0.3532 & 0.4512 & \\
  & Ward & 0.3282 & 0.4232 & 0.4630 & & 0.3387 & 0.4327 & 0.4725 & & 0.3392 & 0.4330 & 0.4727 & \\
\hline
\multirow{3}{*}{Quora}
  & EigenLI & \textbf{0.8389} & \textbf{0.8544} & 0.8287 & \multirow{3}{*}{0.8441} & \textbf{0.8408} & \textbf{0.8557} & 0.8306 & \multirow{3}{*}{0.8457} & \textbf{0.8408} & \textbf{0.8557} & 0.8306 & \multirow{3}{*}{0.8458} \\
  & KMeans++ & 0.8190 & 0.8439 & \textbf{0.8442} & & 0.8211 & 0.8455 & \textbf{0.8458} & & 0.8212 & 0.8456 & \textbf{0.8459} & \\
  & Ward & 0.8337 & 0.8446 & 0.8442 & & 0.8356 & 0.8462 & 0.8458 & & 0.8356 & 0.8462 & 0.8459 & \\
\hline
\multirow{3}{*}{SCIDOCS}
  & EigenLI & \textbf{0.2086} & 0.2520 & \textbf{0.2846} & \multirow{3}{*}{0.2740} & \textbf{0.2211} & 0.2642 & \textbf{0.2931} & \multirow{3}{*}{0.2840} & \textbf{0.2218} & 0.2649 & \textbf{0.2936} & \multirow{3}{*}{0.2846} \\
  & KMeans++ & 0.1881 & 0.2002 & 0.2507 & & 0.2005 & 0.2121 & 0.2627 & & 0.2015 & 0.2129 & 0.2635 & \\
  & Ward & 0.2018 & \textbf{0.2550} & 0.2521 & & 0.2136 & \textbf{0.2663} & 0.2620 & & 0.2147 & \textbf{0.2672} & 0.2627 & \\
\hline
\multirow{3}{*}{SciFact}
  & EigenLI & \textbf{0.5573} & \textbf{0.6128} & \textbf{0.6593} & \multirow{3}{*}{0.6682} & \textbf{0.5628} & \textbf{0.6174} & \textbf{0.6626} & \multirow{3}{*}{0.6723} & \textbf{0.5632} & \textbf{0.6177} & \textbf{0.6629} & \multirow{3}{*}{0.6725} \\
  & KMeans++ & 0.5100 & 0.5533 & 0.6005 & & 0.5174 & 0.5595 & 0.6062 & & 0.5177 & 0.5597 & 0.6065 & \\
  & Ward & 0.5106 & 0.5735 & 0.6073 & & 0.5175 & 0.5790 & 0.6128 & & 0.5179 & 0.5792 & 0.6130 & \\
\hline
\multirow{3}{*}{TREC-COVID}
  & EigenLI & \textbf{0.8289} & \textbf{0.8518} & \textbf{0.9407} & \multirow{3}{*}{0.9056} & \textbf{0.8303} & \textbf{0.8518} & \textbf{0.9407} & \multirow{3}{*}{0.9056} & \textbf{0.8303} & \textbf{0.8518} & \textbf{0.9407} & \multirow{3}{*}{0.9056} \\
  & KMeans++ & 0.6389 & 0.7459 & 0.8129 & & 0.6463 & 0.7504 & 0.8147 & & 0.6463 & 0.7504 & 0.8147 & \\
  & Ward & 0.6930 & 0.7730 & 0.8237 & & 0.6986 & 0.7758 & 0.8237 & & 0.6986 & 0.7758 & 0.8237 & \\
\hline
\multirow{3}{*}{Webis-Touche2020}
  & EigenLI & \textbf{0.4123} & \textbf{0.4752} & \textbf{0.5739} & \multirow{3}{*}{0.4985} & \textbf{0.4182} & \textbf{0.4834} & \textbf{0.5810} & \multirow{3}{*}{0.5096} & \textbf{0.4182} & \textbf{0.4834} & \textbf{0.5810} & \multirow{3}{*}{0.5096} \\
  & KMeans++ & 0.2478 & 0.3167 & 0.3740 & & 0.2621 & 0.3369 & 0.3875 & & 0.2625 & 0.3369 & 0.3875 & \\
  & Ward & 0.3099 & 0.3884 & 0.4405 & & 0.3222 & 0.3999 & 0.4473 & & 0.3226 & 0.3999 & 0.4473 & \\
\bottomrule
\end{tabular}
}
\end{table}
\end{landscape}

\clearpage
\section{Detailed Anisotropic-Model BEIR Results}
\label{sec:anisotropic_appendix}

We evaluate AnswerAI-ColBERT-small and GTE-ModernColBERT with both uncentered (U) and centered (C) representations. Recall that for a document or query with
valid token vectors $\{x_i\}_{i=1}^{n}$, centering computes $\bar{x}=n^{-1}\sum_{i=1}^{n}x_i$ and replaces every valid token vector with
$x_i-\bar{x}$. Documents and queries are centered independently, padding tokens are excluded from the mean, and the centered vectors are not
renormalized. Document centering is performed during offline encoding, while query centering is performed at query-encoding time; consequently, the means need not be
stored after the corresponding representations have been constructed.

\Cref{tab:anisotropic_centering_at10} summarizes the effect of centering on
EigenLI-SV-32 and MUVERA at cutoff 10 for both models. 
\Cref{tab:answerai_all,tab:gte_all} retain the complete comparison with full
MaxSim, EigenLI-MV, EigenLI-SV, $k$-means++, Ward, and MUVERA at cutoffs 10,
100, and 1000. Only the MUVERA rows use the centered representations described
above; the other rows are unchanged.

\begin{landscape}
\begin{table}[p]
\centering
\scriptsize
\setlength{\tabcolsep}{1.8pt}
\renewcommand{\arraystretch}{0.88}
\caption{Effect of centering on EigenLI-SV-32 and MUVERA for
AnswerAI-ColBERT-small and GTE-ModernColBERT. U and C denote uncentered and
centered results, respectively, and $\Delta=C-U$.}
\label{tab:anisotropic_centering_at10}
\resizebox{\linewidth}{!}{%
\begin{tabular}{ll|rrr|rrr|rrr||rrr|rrr|rrr}
\toprule
\multirow{3}{*}{Dataset} & \multirow{3}{*}{Method}
& \multicolumn{9}{c||}{\textbf{AnswerAI-ColBERT-small}}
& \multicolumn{9}{c}{\textbf{GTE-ModernColBERT}} \\
\cmidrule(lr){3-11}\cmidrule(lr){12-20}
& & \multicolumn{3}{c|}{Recall@10}
& \multicolumn{3}{c|}{nDCG@10}
& \multicolumn{3}{c||}{MRR@10}
& \multicolumn{3}{c|}{Recall@10}
& \multicolumn{3}{c|}{nDCG@10}
& \multicolumn{3}{c}{MRR@10} \\
\cmidrule(lr){3-5}\cmidrule(lr){6-8}\cmidrule(lr){9-11}
\cmidrule(lr){12-14}\cmidrule(lr){15-17}\cmidrule(lr){18-20}
& & U & C & $\Delta$ & U & C & $\Delta$ & U & C & $\Delta$
& U & C & $\Delta$ & U & C & $\Delta$ & U & C & $\Delta$ \\
\midrule
ArguAna & EigenLI-SV-32 & 0.6400 & 0.6300 & -0.0100 & 0.4339 & 0.4249 & -0.0090 & 0.3696 & 0.3613 & -0.0082 & 0.7167 & 0.6667 & -0.0500 & 0.4791 & 0.4415 & -0.0376 & 0.4042 & 0.3699 & -0.0343 \\
& MUVERA & 0.1167 & 0.5933 & +0.4767 & 0.0488 & 0.3837 & +0.3349 & 0.0283 & 0.3187 & +0.2905 & 0.0267 & 0.6467 & +0.6200 & 0.0111 & 0.4233 & +0.4122 & 0.0065 & 0.3536 & +0.3471 \\
\midrule
Climate-FEVER & EigenLI-SV-32 & 0.3027 & 0.2442 & -0.0585 & 0.2351 & 0.1843 & -0.0508 & 0.2784 & 0.2335 & -0.0449 & 0.3629 & 0.2388 & -0.1241 & 0.2957 & 0.1784 & -0.1173 & 0.3689 & 0.2165 & -0.1524 \\
& MUVERA & 0.0022 & 0.0485 & +0.0463 & 0.0011 & 0.0360 & +0.0349 & 0.0011 & 0.0450 & +0.0439 & 0.0000 & 0.0381 & +0.0381 & 0.0000 & 0.0305 & +0.0305 & 0.0000 & 0.0478 & +0.0478 \\
\midrule
DBPedia-Entity & EigenLI-SV-32 & 0.4333 & 0.4291 & -0.0042 & 0.4567 & 0.4488 & -0.0079 & 0.7399 & 0.7499 & +0.0099 & 0.3605 & 0.4237 & +0.0632 & 0.3752 & 0.4191 & +0.0440 & 0.6324 & 0.6925 & +0.0601 \\
& MUVERA & 0.0104 & 0.1684 & +0.1580 & 0.0054 & 0.1639 & +0.1585 & 0.0082 & 0.3448 & +0.3366 & 0.0512 & 0.1216 & +0.0703 & 0.0535 & 0.1125 & +0.0590 & 0.1447 & 0.2319 & +0.0873 \\
\midrule
FEVER & EigenLI-SV-32 & 0.8937 & 0.8766 & -0.0171 & 0.7789 & 0.7371 & -0.0418 & 0.7702 & 0.7177 & -0.0525 & 0.8779 & 0.8192 & -0.0587 & 0.7464 & 0.6746 & -0.0718 & 0.7304 & 0.6521 & -0.0783 \\
& MUVERA & 0.0004 & 0.3550 & +0.3547 & 0.0001 & 0.2517 & +0.2515 & 0.0001 & 0.2294 & +0.2293 & 0.0002 & 0.2502 & +0.2500 & 0.0001 & 0.1829 & +0.1828 & 0.0000 & 0.1687 & +0.1687 \\
\midrule
FiQA & EigenLI-SV-32 & 0.4010 & 0.4179 & +0.0169 & 0.3405 & 0.3589 & +0.0185 & 0.4155 & 0.4425 & +0.0270 & 0.4722 & 0.4573 & -0.0150 & 0.4052 & 0.3987 & -0.0065 & 0.4842 & 0.4830 & -0.0012 \\
& MUVERA & 0.0002 & 0.2110 & +0.2108 & 0.0002 & 0.1589 & +0.1587 & 0.0004 & 0.1939 & +0.1935 & 0.0094 & 0.2376 & +0.2282 & 0.0053 & 0.1818 & +0.1765 & 0.0056 & 0.2137 & +0.2081 \\
\midrule
HotpotQA & EigenLI-SV-32 & 0.6980 & 0.6910 & -0.0070 & 0.6757 & 0.6760 & +0.0003 & 0.8479 & 0.8653 & +0.0174 & 0.6792 & 0.6982 & +0.0190 & 0.6501 & 0.6734 & +0.0232 & 0.8148 & 0.8475 & +0.0327 \\
& MUVERA & 0.0027 & 0.4988 & +0.4961 & 0.0017 & 0.4750 & +0.4733 & 0.0020 & 0.6404 & +0.6384 & 0.0002 & 0.4684 & +0.4682 & 0.0001 & 0.4531 & +0.4530 & 0.0000 & 0.6186 & +0.6186 \\
\midrule
MS MARCO & EigenLI-SV-32 & 0.6168 & 0.6114 & -0.0054 & 0.4043 & 0.3966 & -0.0077 & 0.3410 & 0.3328 & -0.0082 & 0.6131 & 0.6131 & -0.0001 & 0.3998 & 0.4019 & +0.0021 & 0.3358 & 0.3388 & +0.0030 \\
& MUVERA & 0.0069 & 0.3999 & +0.3930 & 0.0041 & 0.2438 & +0.2397 & 0.0033 & 0.1979 & +0.1947 & 0.0229 & 0.3678 & +0.3449 & 0.0131 & 0.2322 & +0.2190 & 0.0103 & 0.1928 & +0.1825 \\
\midrule
NQ & EigenLI-SV-32 & 0.6611 & 0.6850 & +0.0239 & 0.4592 & 0.4755 & +0.0164 & 0.4089 & 0.4226 & +0.0137 & 0.7222 & 0.6956 & -0.0267 & 0.4983 & 0.4809 & -0.0174 & 0.4467 & 0.4292 & -0.0175 \\
& MUVERA & 0.0000 & 0.2439 & +0.2439 & 0.0000 & 0.1420 & +0.1420 & 0.0000 & 0.1158 & +0.1158 & 0.0000 & 0.3300 & +0.3300 & 0.0000 & 0.2010 & +0.2010 & 0.0000 & 0.1709 & +0.1709 \\
\midrule
Quora & EigenLI-SV-32 & 0.9069 & 0.8766 & -0.0304 & 0.8205 & 0.7834 & -0.0371 & 0.8126 & 0.7740 & -0.0387 & 0.8428 & 0.7012 & -0.1416 & 0.7048 & 0.5637 & -0.1411 & 0.6776 & 0.5367 & -0.1408 \\
& MUVERA & 0.0512 & 0.9117 & +0.8605 & 0.0369 & 0.8194 & +0.7825 & 0.0351 & 0.8129 & +0.7778 & 0.0918 & 0.7648 & +0.6730 & 0.0564 & 0.6187 & +0.5623 & 0.0488 & 0.5922 & +0.5434 \\
\midrule
SCIDOCS & EigenLI-SV-32 & 0.1495 & 0.1495 & +0.0000 & 0.1382 & 0.1376 & -0.0006 & 0.2369 & 0.2314 & -0.0055 & 0.1635 & 0.1435 & -0.0200 & 0.1561 & 0.1451 & -0.0110 & 0.2778 & 0.2845 & +0.0068 \\
& MUVERA & 0.0185 & 0.0745 & +0.0560 & 0.0149 & 0.0663 & +0.0514 & 0.0258 & 0.1120 & +0.0862 & 0.0080 & 0.1005 & +0.0925 & 0.0074 & 0.1004 & +0.0930 & 0.0163 & 0.1936 & +0.1773 \\
\midrule
SciFact & EigenLI-SV-32 & 0.7978 & 0.8163 & +0.0186 & 0.6777 & 0.7016 & +0.0239 & 0.6431 & 0.6740 & +0.0309 & 0.8424 & 0.8297 & -0.0127 & 0.7177 & 0.7204 & +0.0027 & 0.6871 & 0.6909 & +0.0038 \\
& MUVERA & 0.1673 & 0.7643 & +0.5969 & 0.0950 & 0.6198 & +0.5248 & 0.0750 & 0.5843 & +0.5093 & 0.0783 & 0.7573 & +0.6789 & 0.0415 & 0.6074 & +0.5659 & 0.0322 & 0.5688 & +0.5366 \\
\midrule
TREC-COVID & EigenLI-SV-32 & 0.7380 & 0.7960 & +0.0580 & 0.7143 & 0.7479 & +0.0336 & 0.9290 & 0.9100 & -0.0190 & 0.7220 & 0.8000 & +0.0780 & 0.6593 & 0.7406 & +0.0813 & 0.8206 & 0.8945 & +0.0739 \\
& MUVERA & 0.0040 & 0.3500 & +0.3460 & 0.0057 & 0.3398 & +0.3342 & 0.0220 & 0.5888 & +0.5668 & 0.0280 & 0.5560 & +0.5280 & 0.0243 & 0.5174 & +0.4931 & 0.0900 & 0.7036 & +0.6135 \\
\midrule
Webis-Touche2020 & EigenLI-SV-32 & 0.2035 & 0.2657 & +0.0622 & 0.2026 & 0.2593 & +0.0568 & 0.4598 & 0.5291 & +0.0693 & 0.2915 & 0.3231 & +0.0316 & 0.2965 & 0.3401 & +0.0436 & 0.5279 & 0.5936 & +0.0657 \\
& MUVERA & 0.0122 & 0.0297 & +0.0174 & 0.0127 & 0.0257 & +0.0130 & 0.0383 & 0.0700 & +0.0317 & 0.0041 & 0.1425 & +0.1384 & 0.0060 & 0.1463 & +0.1404 & 0.0204 & 0.3841 & +0.3637 \\
\midrule
\bottomrule
\end{tabular}%
}
\end{table}
\end{landscape}

\clearpage

\begingroup
\scriptsize
\setlength{\tabcolsep}{2.5pt}
\begin{longtable}{llc|rrr|rrr|rrr}
\caption{AnswerAI-ColBERT-small BEIR effectiveness results; results use centered MUVERA.}\label{tab:answerai_all} \\
\toprule
Dataset & Method & $k$ & R@10 & R@100 & R@1000 & n@10 & n@100 & n@1000 & M@10 & M@100 & M@1000 \\
\midrule
\endfirsthead
\multicolumn{12}{c}{\tablename\ \thetable{} -- continued} \\
\toprule
Dataset & Method & $k$ & R@10 & R@100 & R@1000 & n@10 & n@100 & n@1000 & M@10 & M@100 & M@1000 \\
\midrule
\endhead
\midrule
\multicolumn{12}{r}{Continued on next page} \\
\endfoot
\bottomrule
\endlastfoot
ArguAna & Full MaxSim & - & 0.6633 & 0.9100 & 0.9833 & 0.4283 & 0.4820 & 0.4915 & 0.3550 & 0.3665 & 0.3668 \\
\cmidrule{2-12}
& & 8 & 0.6533 & 0.9100 & 0.9867 & 0.4144 & 0.4697 & 0.4794 & 0.3401 & 0.3517 & 0.3520 \\
& EigenLI & 16 & 0.6533 & 0.9267 & 0.9833 & 0.4432 & 0.5031 & 0.5102 & 0.3775 & 0.3904 & 0.3907 \\
& & 32 & 0.6400 & 0.9100 & 0.9833 & 0.4339 & 0.4918 & 0.5009 & 0.3696 & 0.3815 & 0.3818 \\
\cmidrule{2-12}
& & 8 & 0.5167 & 0.8600 & 0.9833 & 0.3359 & 0.4072 & 0.4230 & 0.2811 & 0.2948 & 0.2954 \\
& KMeans++ & 16 & 0.6033 & 0.8967 & 0.9800 & 0.3952 & 0.4598 & 0.4707 & 0.3315 & 0.3455 & 0.3460 \\
& & 32 & 0.6700 & 0.9100 & 0.9900 & 0.4411 & 0.4939 & 0.5041 & 0.3700 & 0.3815 & 0.3819 \\
\cmidrule{2-12}
& & 8 & 0.6033 & 0.9167 & 0.9867 & 0.3888 & 0.4572 & 0.4660 & 0.3222 & 0.3369 & 0.3372 \\
& Ward & 16 & 0.6433 & 0.9233 & 0.9867 & 0.4395 & 0.5001 & 0.5082 & 0.3751 & 0.3877 & 0.3881 \\
& & 32 & 0.7000 & 0.9233 & 0.9900 & 0.4719 & 0.5202 & 0.5285 & 0.4014 & 0.4115 & 0.4118 \\
\cmidrule{2-12}
& MUVERA & - & 0.5933 & 0.8767 & 0.9700 & 0.3837 & 0.4440 & 0.4557 & 0.3187 & 0.3311 & 0.3315 \\
\midrule
Climate-FEVER & Full MaxSim & - & 0.3849 & 0.6712 & 0.8359 & 0.3110 & 0.3924 & 0.4216 & 0.3983 & 0.4092 & 0.4095 \\
\cmidrule{2-12}
& & 8 & 0.3450 & 0.5694 & 0.7281 & 0.2884 & 0.3496 & 0.3772 & 0.3527 & 0.3634 & 0.3636 \\
& EigenLI & 16 & 0.3317 & 0.5442 & 0.7410 & 0.2789 & 0.3368 & 0.3717 & 0.3439 & 0.3537 & 0.3543 \\
& & 32 & 0.3027 & 0.5273 & 0.7212 & 0.2351 & 0.2947 & 0.3283 & 0.2784 & 0.2885 & 0.2892 \\
\cmidrule{2-12}
& & 8 & 0.3281 & 0.5244 & 0.6899 & 0.2698 & 0.3255 & 0.3548 & 0.3308 & 0.3382 & 0.3387 \\
& KMeans++ & 16 & 0.3117 & 0.5142 & 0.6969 & 0.2345 & 0.2910 & 0.3233 & 0.2708 & 0.2806 & 0.2812 \\
& & 32 & 0.3008 & 0.5180 & 0.7082 & 0.2388 & 0.3000 & 0.3333 & 0.2949 & 0.3058 & 0.3064 \\
\cmidrule{2-12}
& & 8 & 0.3572 & 0.5786 & 0.7386 & 0.3082 & 0.3707 & 0.3993 & 0.3871 & 0.3959 & 0.3963 \\
& Ward & 16 & 0.3562 & 0.5765 & 0.7432 & 0.2991 & 0.3612 & 0.3910 & 0.3684 & 0.3786 & 0.3790 \\
& & 32 & 0.3027 & 0.5306 & 0.7256 & 0.2448 & 0.3075 & 0.3412 & 0.2941 & 0.3060 & 0.3065 \\
\cmidrule{2-12}
& MUVERA & - & 0.0485 & 0.1298 & 0.2197 & 0.0360 & 0.0543 & 0.0690 & 0.0450 & 0.0495 & 0.0501 \\
\midrule
DBPedia-Entity & Full MaxSim & - & 0.5318 & 0.6663 & 0.8824 & 0.5427 & 0.5911 & 0.6591 & 0.8408 & 0.8429 & 0.8430 \\
\cmidrule{2-12}
& & 8 & 0.3364 & 0.3963 & 0.6462 & 0.3472 & 0.3734 & 0.4511 & 0.6411 & 0.6431 & 0.6433 \\
& EigenLI & 16 & 0.3182 & 0.4006 & 0.6301 & 0.3304 & 0.3639 & 0.4355 & 0.6203 & 0.6241 & 0.6241 \\
& & 32 & 0.4333 & 0.5599 & 0.8030 & 0.4567 & 0.5047 & 0.5783 & 0.7399 & 0.7432 & 0.7432 \\
\cmidrule{2-12}
& & 8 & 0.2405 & 0.2782 & 0.4298 & 0.2423 & 0.2509 & 0.2992 & 0.5225 & 0.5294 & 0.5296 \\
& KMeans++ & 16 & 0.2814 & 0.3246 & 0.4955 & 0.2838 & 0.2965 & 0.3469 & 0.5613 & 0.5650 & 0.5651 \\
& & 32 & 0.4206 & 0.4984 & 0.6757 & 0.4312 & 0.4548 & 0.5113 & 0.7454 & 0.7476 & 0.7476 \\
\cmidrule{2-12}
& & 8 & 0.3177 & 0.3495 & 0.5586 & 0.3420 & 0.3531 & 0.4156 & 0.6562 & 0.6600 & 0.6603 \\
& Ward & 16 & 0.3019 & 0.3537 & 0.5342 & 0.3111 & 0.3323 & 0.3851 & 0.6296 & 0.6331 & 0.6333 \\
& & 32 & 0.4717 & 0.5463 & 0.7514 & 0.4818 & 0.5045 & 0.5692 & 0.8071 & 0.8071 & 0.8072 \\
\cmidrule{2-12}
& MUVERA & - & 0.1684 & 0.2766 & 0.4897 & 0.1639 & 0.2103 & 0.2746 & 0.3448 & 0.3574 & 0.3576 \\
\midrule
FEVER & Full MaxSim & - & 0.9510 & 0.9720 & 0.9821 & 0.9195 & 0.9254 & 0.9270 & 0.9424 & 0.9427 & 0.9427 \\
\cmidrule{2-12}
& & 8 & 0.9003 & 0.9415 & 0.9614 & 0.8030 & 0.8125 & 0.8156 & 0.8046 & 0.8061 & 0.8062 \\
& EigenLI & 16 & 0.8829 & 0.9381 & 0.9578 & 0.7586 & 0.7715 & 0.7744 & 0.7475 & 0.7498 & 0.7498 \\
& & 32 & 0.8937 & 0.9464 & 0.9644 & 0.7789 & 0.7916 & 0.7943 & 0.7702 & 0.7724 & 0.7724 \\
\cmidrule{2-12}
& & 8 & 0.6811 & 0.8248 & 0.9039 & 0.5226 & 0.5538 & 0.5643 & 0.4921 & 0.4983 & 0.4986 \\
& KMeans++ & 16 & 0.7538 & 0.8677 & 0.9201 & 0.5830 & 0.6084 & 0.6153 & 0.5502 & 0.5552 & 0.5555 \\
& & 32 & 0.8229 & 0.9129 & 0.9441 & 0.6642 & 0.6849 & 0.6892 & 0.6402 & 0.6443 & 0.6444 \\
\cmidrule{2-12}
& & 8 & 0.8898 & 0.9341 & 0.9507 & 0.7824 & 0.7924 & 0.7948 & 0.7804 & 0.7822 & 0.7822 \\
& Ward & 16 & 0.8756 & 0.9310 & 0.9517 & 0.7557 & 0.7686 & 0.7716 & 0.7478 & 0.7503 & 0.7503 \\
& & 32 & 0.8605 & 0.9315 & 0.9555 & 0.7290 & 0.7456 & 0.7490 & 0.7143 & 0.7175 & 0.7176 \\
\cmidrule{2-12}
& MUVERA & - & 0.3550 & 0.5912 & 0.7963 & 0.2517 & 0.3012 & 0.3273 & 0.2294 & 0.2391 & 0.2400 \\
\midrule
FiQA & Full MaxSim & - & 0.4752 & 0.7158 & 0.8720 & 0.3986 & 0.4635 & 0.4889 & 0.4720 & 0.4803 & 0.4806 \\
\cmidrule{2-12}
& & 8 & 0.3571 & 0.6001 & 0.8342 & 0.2946 & 0.3591 & 0.3954 & 0.3643 & 0.3722 & 0.3729 \\
& EigenLI & 16 & 0.3822 & 0.6089 & 0.8175 & 0.3263 & 0.3862 & 0.4193 & 0.4007 & 0.4085 & 0.4092 \\
& & 32 & 0.4010 & 0.6070 & 0.8130 & 0.3405 & 0.3948 & 0.4280 & 0.4155 & 0.4228 & 0.4234 \\
\cmidrule{2-12}
& & 8 & 0.2368 & 0.4372 & 0.7365 & 0.1890 & 0.2411 & 0.2874 & 0.2413 & 0.2500 & 0.2510 \\
& KMeans++ & 16 & 0.2751 & 0.4753 & 0.7004 & 0.2266 & 0.2788 & 0.3143 & 0.2843 & 0.2926 & 0.2933 \\
& & 32 & 0.3307 & 0.5268 & 0.7248 & 0.2773 & 0.3285 & 0.3602 & 0.3446 & 0.3519 & 0.3524 \\
\cmidrule{2-12}
& & 8 & 0.3421 & 0.5687 & 0.8174 & 0.2850 & 0.3456 & 0.3839 & 0.3484 & 0.3563 & 0.3571 \\
& Ward & 16 & 0.3350 & 0.5635 & 0.7770 & 0.2821 & 0.3427 & 0.3762 & 0.3515 & 0.3589 & 0.3595 \\
& & 32 & 0.3637 & 0.5516 & 0.7441 & 0.3050 & 0.3550 & 0.3858 & 0.3720 & 0.3788 & 0.3794 \\
\cmidrule{2-12}
& MUVERA & - & 0.2110 & 0.4030 & 0.6513 & 0.1589 & 0.2080 & 0.2465 & 0.1939 & 0.2019 & 0.2027 \\
\midrule
HotpotQA & Full MaxSim & - & 0.7830 & 0.8842 & 0.9417 & 0.7660 & 0.7923 & 0.8010 & 0.9222 & 0.9229 & 0.9229 \\
\cmidrule{2-12}
& & 8 & 0.6314 & 0.7501 & 0.8482 & 0.6061 & 0.6364 & 0.6512 & 0.7658 & 0.7683 & 0.7685 \\
& EigenLI & 16 & 0.6632 & 0.7815 & 0.8650 & 0.6355 & 0.6660 & 0.6785 & 0.7976 & 0.7997 & 0.7998 \\
& & 32 & 0.6980 & 0.8079 & 0.8852 & 0.6757 & 0.7042 & 0.7158 & 0.8479 & 0.8495 & 0.8496 \\
\cmidrule{2-12}
& & 8 & 0.4723 & 0.6423 & 0.7762 & 0.4343 & 0.4775 & 0.4975 & 0.5661 & 0.5721 & 0.5724 \\
& KMeans++ & 16 & 0.5811 & 0.7173 & 0.8108 & 0.5502 & 0.5855 & 0.5995 & 0.7028 & 0.7066 & 0.7067 \\
& & 32 & 0.7006 & 0.8062 & 0.8771 & 0.6747 & 0.7020 & 0.7127 & 0.8349 & 0.8368 & 0.8369 \\
\cmidrule{2-12}
& & 8 & 0.6209 & 0.7302 & 0.8228 & 0.6001 & 0.6279 & 0.6418 & 0.7624 & 0.7649 & 0.7651 \\
& Ward & 16 & 0.6651 & 0.7702 & 0.8426 & 0.6421 & 0.6693 & 0.6803 & 0.8044 & 0.8065 & 0.8066 \\
& & 32 & 0.7243 & 0.8272 & 0.8927 & 0.7018 & 0.7284 & 0.7383 & 0.8621 & 0.8634 & 0.8634 \\
\cmidrule{2-12}
& MUVERA & - & 0.4988 & 0.6378 & 0.7486 & 0.4750 & 0.5107 & 0.5274 & 0.6404 & 0.6452 & 0.6455 \\
\midrule
MS MARCO & Full MaxSim & - & 0.6564 & 0.9039 & 0.9823 & 0.4352 & 0.4893 & 0.4997 & 0.3692 & 0.3801 & 0.3804 \\
\cmidrule{2-12}
& & 8 & 0.5334 & 0.8175 & 0.9436 & 0.3402 & 0.4012 & 0.4174 & 0.2834 & 0.2954 & 0.2960 \\
& EigenLI & 16 & 0.5411 & 0.8265 & 0.9497 & 0.3452 & 0.4057 & 0.4216 & 0.2874 & 0.2990 & 0.2995 \\
& & 32 & 0.6168 & 0.8774 & 0.9715 & 0.4043 & 0.4609 & 0.4732 & 0.3410 & 0.3522 & 0.3526 \\
\cmidrule{2-12}
& & 8 & 0.3332 & 0.5819 & 0.7763 & 0.2118 & 0.2635 & 0.2879 & 0.1766 & 0.1864 & 0.1872 \\
& KMeans++ & 16 & 0.4298 & 0.7016 & 0.8653 & 0.2720 & 0.3294 & 0.3502 & 0.2259 & 0.2370 & 0.2377 \\
& & 32 & 0.5799 & 0.8350 & 0.9448 & 0.3808 & 0.4357 & 0.4499 & 0.3225 & 0.3334 & 0.3338 \\
\cmidrule{2-12}
& & 8 & 0.5010 & 0.7936 & 0.9317 & 0.3203 & 0.3825 & 0.4003 & 0.2670 & 0.2790 & 0.2797 \\
& Ward & 16 & 0.5072 & 0.7932 & 0.9272 & 0.3195 & 0.3800 & 0.3972 & 0.2637 & 0.2753 & 0.2759 \\
& & 32 & 0.6106 & 0.8675 & 0.9627 & 0.4029 & 0.4585 & 0.4710 & 0.3413 & 0.3523 & 0.3528 \\
\cmidrule{2-12}
& MUVERA & - & 0.3999 & 0.6898 & 0.8739 & 0.2438 & 0.3047 & 0.3282 & 0.1979 & 0.2096 & 0.2104 \\
\midrule
NQ & Full MaxSim & - & 0.8183 & 0.9594 & 0.9967 & 0.5846 & 0.6185 & 0.6236 & 0.5232 & 0.5297 & 0.5298 \\
\cmidrule{2-12}
& & 8 & 0.5383 & 0.8000 & 0.9233 & 0.3594 & 0.4173 & 0.4334 & 0.3164 & 0.3269 & 0.3274 \\
& EigenLI & 16 & 0.5883 & 0.8572 & 0.9533 & 0.3682 & 0.4266 & 0.4390 & 0.3113 & 0.3212 & 0.3215 \\
& & 32 & 0.6611 & 0.9000 & 0.9722 & 0.4592 & 0.5144 & 0.5241 & 0.4089 & 0.4190 & 0.4194 \\
\cmidrule{2-12}
& & 8 & 0.3233 & 0.5722 & 0.7706 & 0.2018 & 0.2542 & 0.2794 & 0.1760 & 0.1846 & 0.1853 \\
& KMeans++ & 16 & 0.4150 & 0.6833 & 0.8383 & 0.2535 & 0.3131 & 0.3337 & 0.2148 & 0.2263 & 0.2269 \\
& & 32 & 0.5267 & 0.7911 & 0.9217 & 0.3780 & 0.4383 & 0.4552 & 0.3474 & 0.3587 & 0.3593 \\
\cmidrule{2-12}
& & 8 & 0.4950 & 0.7328 & 0.8950 & 0.3297 & 0.3818 & 0.4025 & 0.2916 & 0.3011 & 0.3017 \\
& Ward & 16 & 0.5367 & 0.7761 & 0.9122 & 0.3355 & 0.3891 & 0.4069 & 0.2829 & 0.2922 & 0.2928 \\
& & 32 & 0.6006 & 0.8372 & 0.9400 & 0.4211 & 0.4745 & 0.4878 & 0.3782 & 0.3881 & 0.3884 \\
\cmidrule{2-12}
& MUVERA & - & 0.2439 & 0.5511 & 0.7656 & 0.1420 & 0.2084 & 0.2361 & 0.1158 & 0.1285 & 0.1293 \\
\midrule
Quora & Full MaxSim & - & 0.9540 & 0.9947 & 0.9998 & 0.8768 & 0.8892 & 0.8904 & 0.8680 & 0.8691 & 0.8691 \\
\cmidrule{2-12}
& & 8 & 0.9491 & 0.9924 & 0.9992 & 0.8727 & 0.8854 & 0.8868 & 0.8655 & 0.8667 & 0.8667 \\
& EigenLI & 16 & 0.9420 & 0.9924 & 0.9991 & 0.8630 & 0.8778 & 0.8792 & 0.8545 & 0.8560 & 0.8560 \\
& & 32 & 0.9069 & 0.9736 & 0.9931 & 0.8205 & 0.8383 & 0.8416 & 0.8126 & 0.8147 & 0.8148 \\
\cmidrule{2-12}
& & 8 & 0.9392 & 0.9873 & 0.9974 & 0.8544 & 0.8680 & 0.8700 & 0.8453 & 0.8465 & 0.8465 \\
& KMeans++ & 16 & 0.9542 & 0.9947 & 0.9996 & 0.8788 & 0.8911 & 0.8922 & 0.8704 & 0.8715 & 0.8715 \\
& & 32 & 0.9540 & 0.9947 & 0.9998 & 0.8769 & 0.8893 & 0.8904 & 0.8681 & 0.8691 & 0.8691 \\
\cmidrule{2-12}
& & 8 & 0.9474 & 0.9918 & 0.9989 & 0.8680 & 0.8809 & 0.8824 & 0.8596 & 0.8608 & 0.8608 \\
& Ward & 16 & 0.9546 & 0.9950 & 0.9998 & 0.8786 & 0.8909 & 0.8920 & 0.8699 & 0.8710 & 0.8710 \\
& & 32 & 0.9540 & 0.9947 & 0.9998 & 0.8769 & 0.8893 & 0.8904 & 0.8680 & 0.8691 & 0.8691 \\
\cmidrule{2-12}
& MUVERA & - & 0.9117 & 0.9855 & 0.9980 & 0.8194 & 0.8394 & 0.8418 & 0.8129 & 0.8152 & 0.8152 \\
\midrule
SCIDOCS & Full MaxSim & - & 0.1635 & 0.4310 & 0.6908 & 0.1574 & 0.2472 & 0.2994 & 0.2770 & 0.2912 & 0.2916 \\
\cmidrule{2-12}
& & 8 & 0.1310 & 0.3565 & 0.6795 & 0.1165 & 0.1935 & 0.2596 & 0.1932 & 0.2073 & 0.2082 \\
& EigenLI & 16 & 0.1575 & 0.3625 & 0.6600 & 0.1341 & 0.2042 & 0.2652 & 0.2183 & 0.2292 & 0.2302 \\
& & 32 & 0.1495 & 0.3580 & 0.6110 & 0.1382 & 0.2082 & 0.2606 & 0.2369 & 0.2467 & 0.2474 \\
\cmidrule{2-12}
& & 8 & 0.0885 & 0.2955 & 0.6380 & 0.0847 & 0.1551 & 0.2237 & 0.1528 & 0.1711 & 0.1718 \\
& KMeans++ & 16 & 0.1270 & 0.3110 & 0.6130 & 0.1133 & 0.1749 & 0.2362 & 0.1919 & 0.2037 & 0.2044 \\
& & 32 & 0.1245 & 0.3110 & 0.5685 & 0.1168 & 0.1809 & 0.2341 & 0.2081 & 0.2229 & 0.2237 \\
\cmidrule{2-12}
& & 8 & 0.1210 & 0.3445 & 0.6810 & 0.1146 & 0.1889 & 0.2567 & 0.1964 & 0.2120 & 0.2127 \\
& Ward & 16 & 0.1290 & 0.3685 & 0.6690 & 0.1191 & 0.2004 & 0.2604 & 0.2059 & 0.2207 & 0.2211 \\
& & 32 & 0.1430 & 0.3475 & 0.6085 & 0.1373 & 0.2065 & 0.2598 & 0.2504 & 0.2617 & 0.2626 \\
\cmidrule{2-12}
& MUVERA & - & 0.0745 & 0.2515 & 0.4870 & 0.0663 & 0.1255 & 0.1742 & 0.1120 & 0.1265 & 0.1278 \\
\midrule
SciFact & Full MaxSim & - & 0.8446 & 0.9560 & 0.9967 & 0.7446 & 0.7700 & 0.7753 & 0.7200 & 0.7248 & 0.7250 \\
\cmidrule{2-12}
& & 8 & 0.7900 & 0.9360 & 0.9800 & 0.6468 & 0.6789 & 0.6844 & 0.6079 & 0.6136 & 0.6137 \\
& EigenLI & 16 & 0.8034 & 0.9460 & 0.9800 & 0.6767 & 0.7076 & 0.7120 & 0.6410 & 0.6467 & 0.6468 \\
& & 32 & 0.7978 & 0.9277 & 0.9767 & 0.6777 & 0.7068 & 0.7132 & 0.6431 & 0.6483 & 0.6486 \\
\cmidrule{2-12}
& & 8 & 0.7189 & 0.8930 & 0.9700 & 0.5691 & 0.6100 & 0.6197 & 0.5319 & 0.5396 & 0.5399 \\
& KMeans++ & 16 & 0.7783 & 0.9132 & 0.9733 & 0.6356 & 0.6657 & 0.6737 & 0.5975 & 0.6027 & 0.6030 \\
& & 32 & 0.7777 & 0.9170 & 0.9867 & 0.6528 & 0.6823 & 0.6913 & 0.6198 & 0.6244 & 0.6247 \\
\cmidrule{2-12}
& & 8 & 0.7733 & 0.9150 & 0.9767 & 0.6340 & 0.6665 & 0.6745 & 0.5955 & 0.6018 & 0.6021 \\
& Ward & 16 & 0.8094 & 0.9450 & 0.9800 & 0.6699 & 0.6997 & 0.7041 & 0.6290 & 0.6342 & 0.6344 \\
& & 32 & 0.8218 & 0.9450 & 0.9867 & 0.7044 & 0.7311 & 0.7364 & 0.6702 & 0.6748 & 0.6750 \\
\cmidrule{2-12}
& MUVERA & - & 0.7643 & 0.9193 & 0.9760 & 0.6198 & 0.6548 & 0.6620 & 0.5843 & 0.5900 & 0.5903 \\
\midrule
TREC-COVID & Full MaxSim & - & 0.8800 & 0.6730 & 0.5412 & 0.8348 & 0.6470 & 0.5690 & 0.9800 & 0.9800 & 0.9800 \\
\cmidrule{2-12}
& & 8 & 0.5400 & 0.3668 & 0.3090 & 0.5167 & 0.3658 & 0.3294 & 0.7852 & 0.7864 & 0.7864 \\
& EigenLI & 16 & 0.5660 & 0.3864 & 0.3416 & 0.5396 & 0.3824 & 0.3611 & 0.8052 & 0.8052 & 0.8052 \\
& & 32 & 0.7380 & 0.5490 & 0.4406 & 0.7143 & 0.5324 & 0.4674 & 0.9290 & 0.9290 & 0.9290 \\
\cmidrule{2-12}
& & 8 & 0.4740 & 0.2986 & 0.2546 & 0.4468 & 0.3004 & 0.2711 & 0.7328 & 0.7397 & 0.7397 \\
& KMeans++ & 16 & 0.4960 & 0.3318 & 0.2748 & 0.4701 & 0.3295 & 0.2947 & 0.7125 & 0.7125 & 0.7125 \\
& & 32 & 0.6080 & 0.4198 & 0.3472 & 0.6022 & 0.4182 & 0.3717 & 0.8773 & 0.8789 & 0.8789 \\
\cmidrule{2-12}
& & 8 & 0.5320 & 0.3788 & 0.3128 & 0.5172 & 0.3766 & 0.3346 & 0.8073 & 0.8080 & 0.8080 \\
& Ward & 16 & 0.5740 & 0.3994 & 0.3228 & 0.5472 & 0.3939 & 0.3481 & 0.7766 & 0.7766 & 0.7766 \\
& & 32 & 0.6340 & 0.4554 & 0.3783 & 0.6189 & 0.4503 & 0.4043 & 0.8653 & 0.8670 & 0.8670 \\
\cmidrule{2-12}
& MUVERA & - & 0.3500 & 0.1538 & 0.0967 & 0.3398 & 0.1722 & 0.1180 & 0.5888 & 0.5923 & 0.5923 \\
\midrule
Webis-Touche2020 & Full MaxSim & - & 0.2508 & 0.5070 & 0.8579 & 0.2609 & 0.3954 & 0.5093 & 0.4999 & 0.5109 & 0.5109 \\
\cmidrule{2-12}
& & 8 & 0.1142 & 0.3124 & 0.6593 & 0.1015 & 0.2080 & 0.3189 & 0.1874 & 0.2035 & 0.2040 \\
& EigenLI & 16 & 0.1407 & 0.3314 & 0.6933 & 0.1336 & 0.2312 & 0.3457 & 0.2745 & 0.2919 & 0.2923 \\
& & 32 & 0.2035 & 0.4221 & 0.7540 & 0.2026 & 0.3139 & 0.4240 & 0.4598 & 0.4677 & 0.4677 \\
\cmidrule{2-12}
& & 8 & 0.0673 & 0.2132 & 0.5573 & 0.0638 & 0.1338 & 0.2450 & 0.1336 & 0.1530 & 0.1539 \\
& KMeans++ & 16 & 0.0879 & 0.2432 & 0.5555 & 0.0851 & 0.1572 & 0.2598 & 0.1993 & 0.2183 & 0.2188 \\
& & 32 & 0.1427 & 0.3210 & 0.6372 & 0.1354 & 0.2184 & 0.3211 & 0.3342 & 0.3467 & 0.3468 \\
\cmidrule{2-12}
& & 8 & 0.1087 & 0.2960 & 0.6421 & 0.0944 & 0.1928 & 0.3053 & 0.1698 & 0.1822 & 0.1833 \\
& Ward & 16 & 0.1062 & 0.2862 & 0.6032 & 0.0998 & 0.1900 & 0.2929 & 0.2150 & 0.2339 & 0.2347 \\
& & 32 & 0.1728 & 0.3380 & 0.6725 & 0.1676 & 0.2449 & 0.3537 & 0.3974 & 0.4096 & 0.4098 \\
\cmidrule{2-12}
& MUVERA & - & 0.0297 & 0.1026 & 0.3689 & 0.0257 & 0.0585 & 0.1361 & 0.0700 & 0.0891 & 0.0902 \\
\midrule
\end{longtable}
\endgroup

\clearpage

\begingroup
\scriptsize
\setlength{\tabcolsep}{2.5pt}
\begin{longtable}{llc|rrr|rrr|rrr}
\caption{GTE-ModernColBERT BEIR effectiveness results; results use centered MUVERA.}\label{tab:gte_all} \\
\toprule
Dataset & Method & $k$ & R@10 & R@100 & R@1000 & n@10 & n@100 & n@1000 & M@10 & M@100 & M@1000 \\
\midrule
\endfirsthead
\multicolumn{12}{c}{\tablename\ \thetable{} -- continued} \\
\toprule
Dataset & Method & $k$ & R@10 & R@100 & R@1000 & n@10 & n@100 & n@1000 & M@10 & M@100 & M@1000 \\
\midrule
\endhead
\midrule
\multicolumn{12}{r}{Continued on next page} \\
\endfoot
\bottomrule
\endlastfoot
ArguAna & Full MaxSim & - & 0.6933 & 0.9533 & 1.0000 & 0.4474 & 0.5048 & 0.5105 & 0.3700 & 0.3825 & 0.3827 \\
\cmidrule{2-12}
& & 8 & 0.7467 & 0.9633 & 0.9967 & 0.4799 & 0.5280 & 0.5320 & 0.3969 & 0.4076 & 0.4077 \\
& EigenLI & 16 & 0.7433 & 0.9600 & 0.9967 & 0.5030 & 0.5506 & 0.5554 & 0.4273 & 0.4375 & 0.4377 \\
& & 32 & 0.7167 & 0.9633 & 1.0000 & 0.4791 & 0.5323 & 0.5370 & 0.4042 & 0.4151 & 0.4153 \\
\cmidrule{2-12}
& & 8 & 0.6033 & 0.9267 & 0.9833 & 0.3700 & 0.4405 & 0.4480 & 0.2981 & 0.3131 & 0.3135 \\
& KMeans++ & 16 & 0.7100 & 0.9400 & 0.9867 & 0.4382 & 0.4874 & 0.4934 & 0.3539 & 0.3639 & 0.3641 \\
& & 32 & 0.7233 & 0.9600 & 0.9933 & 0.4705 & 0.5238 & 0.5280 & 0.3915 & 0.4036 & 0.4038 \\
\cmidrule{2-12}
& & 8 & 0.7100 & 0.9500 & 0.9933 & 0.4470 & 0.4998 & 0.5052 & 0.3654 & 0.3769 & 0.3771 \\
& Ward & 16 & 0.7367 & 0.9667 & 0.9933 & 0.4856 & 0.5359 & 0.5392 & 0.4069 & 0.4177 & 0.4178 \\
& & 32 & 0.7267 & 0.9567 & 0.9933 & 0.4874 & 0.5390 & 0.5436 & 0.4115 & 0.4230 & 0.4232 \\
\cmidrule{2-12}
& MUVERA & - & 0.6467 & 0.8967 & 0.9767 & 0.4233 & 0.4788 & 0.4886 & 0.3536 & 0.3659 & 0.3662 \\
\midrule
Climate-FEVER & Full MaxSim & - & 0.4637 & 0.7009 & 0.8767 & 0.3873 & 0.4555 & 0.4869 & 0.4703 & 0.4788 & 0.4790 \\
\cmidrule{2-12}
& & 8 & 0.3622 & 0.5975 & 0.7912 & 0.2891 & 0.3573 & 0.3906 & 0.3593 & 0.3695 & 0.3700 \\
& EigenLI & 16 & 0.3411 & 0.5870 & 0.7783 & 0.2753 & 0.3446 & 0.3783 & 0.3421 & 0.3547 & 0.3552 \\
& & 32 & 0.3629 & 0.5930 & 0.7910 & 0.2957 & 0.3610 & 0.3957 & 0.3689 & 0.3791 & 0.3796 \\
\cmidrule{2-12}
& & 8 & 0.3113 & 0.5623 & 0.7699 & 0.2166 & 0.2863 & 0.3221 & 0.2370 & 0.2472 & 0.2478 \\
& KMeans++ & 16 & 0.3704 & 0.6106 & 0.7891 & 0.2789 & 0.3473 & 0.3797 & 0.3252 & 0.3333 & 0.3338 \\
& & 32 & 0.3622 & 0.6187 & 0.8057 & 0.2770 & 0.3506 & 0.3839 & 0.3318 & 0.3414 & 0.3418 \\
\cmidrule{2-12}
& & 8 & 0.3798 & 0.6237 & 0.8178 & 0.2976 & 0.3688 & 0.4030 & 0.3664 & 0.3752 & 0.3756 \\
& Ward & 16 & 0.3932 & 0.6472 & 0.8221 & 0.3090 & 0.3816 & 0.4122 & 0.3681 & 0.3771 & 0.3775 \\
& & 32 & 0.4257 & 0.6529 & 0.8138 & 0.3401 & 0.4065 & 0.4352 & 0.4040 & 0.4115 & 0.4117 \\
\cmidrule{2-12}
& MUVERA & - & 0.0381 & 0.0980 & 0.1547 & 0.0305 & 0.0450 & 0.0533 & 0.0478 & 0.0511 & 0.0515 \\
\midrule
DBPedia-Entity & Full MaxSim & - & 0.5576 & 0.6864 & 0.8835 & 0.5601 & 0.6120 & 0.6733 & 0.8591 & 0.8596 & 0.8597 \\
\cmidrule{2-12}
& & 8 & 0.4566 & 0.5643 & 0.7908 & 0.4541 & 0.4952 & 0.5684 & 0.7746 & 0.7754 & 0.7755 \\
& EigenLI & 16 & 0.4136 & 0.5903 & 0.8165 & 0.4327 & 0.5023 & 0.5760 & 0.7120 & 0.7158 & 0.7160 \\
& & 32 & 0.3605 & 0.5450 & 0.7887 & 0.3752 & 0.4548 & 0.5283 & 0.6324 & 0.6348 & 0.6350 \\
\cmidrule{2-12}
& & 8 & 0.4038 & 0.5226 & 0.7488 & 0.4117 & 0.4541 & 0.5244 & 0.7173 & 0.7186 & 0.7186 \\
& KMeans++ & 16 & 0.4471 & 0.5956 & 0.8255 & 0.4704 & 0.5249 & 0.6001 & 0.8132 & 0.8153 & 0.8153 \\
& & 32 & 0.5231 & 0.6454 & 0.8647 & 0.5314 & 0.5796 & 0.6489 & 0.8050 & 0.8054 & 0.8054 \\
\cmidrule{2-12}
& & 8 & 0.4435 & 0.5645 & 0.8024 & 0.4557 & 0.4973 & 0.5724 & 0.7568 & 0.7581 & 0.7581 \\
& Ward & 16 & 0.4821 & 0.6314 & 0.8569 & 0.5023 & 0.5629 & 0.6327 & 0.8363 & 0.8368 & 0.8368 \\
& & 32 & 0.5246 & 0.6611 & 0.8719 & 0.5228 & 0.5745 & 0.6434 & 0.8065 & 0.8069 & 0.8070 \\
\cmidrule{2-12}
& MUVERA & - & 0.1216 & 0.1757 & 0.3260 & 0.1125 & 0.1375 & 0.1813 & 0.2319 & 0.2426 & 0.2433 \\
\midrule
FEVER & Full MaxSim & - & 0.9508 & 0.9725 & 0.9839 & 0.9206 & 0.9266 & 0.9285 & 0.9447 & 0.9448 & 0.9448 \\
\cmidrule{2-12}
& & 8 & 0.8613 & 0.9444 & 0.9667 & 0.6896 & 0.7095 & 0.7129 & 0.6610 & 0.6649 & 0.6650 \\
& EigenLI & 16 & 0.8425 & 0.9432 & 0.9695 & 0.6842 & 0.7078 & 0.7117 & 0.6574 & 0.6618 & 0.6619 \\
& & 32 & 0.8779 & 0.9505 & 0.9726 & 0.7464 & 0.7640 & 0.7674 & 0.7304 & 0.7335 & 0.7336 \\
\cmidrule{2-12}
& & 8 & 0.7068 & 0.8470 & 0.8972 & 0.5362 & 0.5684 & 0.5752 & 0.5021 & 0.5088 & 0.5090 \\
& KMeans++ & 16 & 0.7878 & 0.9237 & 0.9563 & 0.6106 & 0.6423 & 0.6469 & 0.5769 & 0.5834 & 0.5835 \\
& & 32 & 0.8559 & 0.9487 & 0.9692 & 0.7093 & 0.7313 & 0.7344 & 0.6902 & 0.6942 & 0.6942 \\
\cmidrule{2-12}
& & 8 & 0.8781 & 0.9502 & 0.9647 & 0.6889 & 0.7063 & 0.7085 & 0.6556 & 0.6588 & 0.6589 \\
& Ward & 16 & 0.8793 & 0.9534 & 0.9697 & 0.7274 & 0.7454 & 0.7479 & 0.7067 & 0.7099 & 0.7099 \\
& & 32 & 0.9280 & 0.9623 & 0.9744 & 0.8286 & 0.8374 & 0.8393 & 0.8288 & 0.8298 & 0.8298 \\
\cmidrule{2-12}
& MUVERA & - & 0.2502 & 0.3921 & 0.5751 & 0.1829 & 0.2122 & 0.2350 & 0.1687 & 0.1743 & 0.1750 \\
\midrule
FiQA & Full MaxSim & - & 0.5204 & 0.7546 & 0.9171 & 0.4547 & 0.5177 & 0.5446 & 0.5401 & 0.5466 & 0.5469 \\
\cmidrule{2-12}
& & 8 & 0.4184 & 0.6673 & 0.8673 & 0.3533 & 0.4188 & 0.4514 & 0.4293 & 0.4376 & 0.4382 \\
& EigenLI & 16 & 0.4658 & 0.7163 & 0.8873 & 0.4002 & 0.4654 & 0.4931 & 0.4727 & 0.4806 & 0.4810 \\
& & 32 & 0.4722 & 0.7229 & 0.8892 & 0.4052 & 0.4705 & 0.4980 & 0.4842 & 0.4916 & 0.4921 \\
\cmidrule{2-12}
& & 8 & 0.2980 & 0.5637 & 0.8150 & 0.2422 & 0.3108 & 0.3507 & 0.2931 & 0.3036 & 0.3043 \\
& KMeans++ & 16 & 0.3945 & 0.6268 & 0.8446 & 0.3197 & 0.3796 & 0.4142 & 0.3802 & 0.3876 & 0.3882 \\
& & 32 & 0.4464 & 0.6882 & 0.8791 & 0.3760 & 0.4398 & 0.4707 & 0.4469 & 0.4546 & 0.4551 \\
\cmidrule{2-12}
& & 8 & 0.3713 & 0.6249 & 0.8409 & 0.3020 & 0.3689 & 0.4036 & 0.3663 & 0.3762 & 0.3768 \\
& Ward & 16 & 0.4458 & 0.7035 & 0.8821 & 0.3720 & 0.4400 & 0.4687 & 0.4419 & 0.4504 & 0.4509 \\
& & 32 & 0.4962 & 0.7341 & 0.8996 & 0.4205 & 0.4848 & 0.5119 & 0.4916 & 0.5003 & 0.5007 \\
\cmidrule{2-12}
& MUVERA & - & 0.2376 & 0.4353 & 0.6854 & 0.1818 & 0.2325 & 0.2714 & 0.2137 & 0.2226 & 0.2234 \\
\midrule
HotpotQA & Full MaxSim & - & 0.8014 & 0.8974 & 0.9518 & 0.7870 & 0.8121 & 0.8204 & 0.9386 & 0.9392 & 0.9392 \\
\cmidrule{2-12}
& & 8 & 0.5864 & 0.7233 & 0.8244 & 0.5547 & 0.5898 & 0.6049 & 0.7100 & 0.7137 & 0.7139 \\
& EigenLI & 16 & 0.6436 & 0.7817 & 0.8730 & 0.6091 & 0.6445 & 0.6583 & 0.7639 & 0.7668 & 0.7670 \\
& & 32 & 0.6792 & 0.8171 & 0.8993 & 0.6501 & 0.6855 & 0.6980 & 0.8148 & 0.8170 & 0.8170 \\
\cmidrule{2-12}
& & 8 & 0.5418 & 0.6797 & 0.7874 & 0.5067 & 0.5420 & 0.5581 & 0.6521 & 0.6562 & 0.6564 \\
& KMeans++ & 16 & 0.6209 & 0.7483 & 0.8394 & 0.5937 & 0.6264 & 0.6401 & 0.7522 & 0.7553 & 0.7555 \\
& & 32 & 0.7071 & 0.8207 & 0.8915 & 0.6853 & 0.7149 & 0.7256 & 0.8496 & 0.8513 & 0.8514 \\
\cmidrule{2-12}
& & 8 & 0.6126 & 0.7322 & 0.8205 & 0.5848 & 0.6158 & 0.6290 & 0.7439 & 0.7471 & 0.7472 \\
& Ward & 16 & 0.6782 & 0.7918 & 0.8673 & 0.6549 & 0.6844 & 0.6958 & 0.8156 & 0.8177 & 0.8178 \\
& & 32 & 0.7356 & 0.8403 & 0.9090 & 0.7157 & 0.7429 & 0.7533 & 0.8776 & 0.8789 & 0.8789 \\
\cmidrule{2-12}
& MUVERA & - & 0.4684 & 0.5934 & 0.6961 & 0.4531 & 0.4847 & 0.5000 & 0.6186 & 0.6234 & 0.6237 \\
\midrule
MS MARCO & Full MaxSim & - & 0.6864 & 0.9314 & 0.9900 & 0.4612 & 0.5147 & 0.5225 & 0.3938 & 0.4045 & 0.4047 \\
\cmidrule{2-12}
& & 8 & 0.5783 & 0.8664 & 0.9747 & 0.3716 & 0.4332 & 0.4474 & 0.3104 & 0.3223 & 0.3228 \\
& EigenLI & 16 & 0.6138 & 0.8916 & 0.9826 & 0.4021 & 0.4620 & 0.4740 & 0.3391 & 0.3510 & 0.3514 \\
& & 32 & 0.6131 & 0.8828 & 0.9811 & 0.3998 & 0.4576 & 0.4705 & 0.3358 & 0.3470 & 0.3475 \\
\cmidrule{2-12}
& & 8 & 0.4854 & 0.7847 & 0.9351 & 0.3119 & 0.3758 & 0.3950 & 0.2610 & 0.2735 & 0.2742 \\
& KMeans++ & 16 & 0.5925 & 0.8698 & 0.9753 & 0.3866 & 0.4462 & 0.4600 & 0.3257 & 0.3374 & 0.3379 \\
& & 32 & 0.6646 & 0.9172 & 0.9870 & 0.4404 & 0.4953 & 0.5047 & 0.3733 & 0.3843 & 0.3846 \\
\cmidrule{2-12}
& & 8 & 0.5717 & 0.8631 & 0.9711 & 0.3689 & 0.4314 & 0.4456 & 0.3091 & 0.3214 & 0.3219 \\
& Ward & 16 & 0.6430 & 0.9055 & 0.9834 & 0.4223 & 0.4791 & 0.4894 & 0.3567 & 0.3678 & 0.3681 \\
& & 32 & 0.6784 & 0.9222 & 0.9887 & 0.4501 & 0.5032 & 0.5122 & 0.3821 & 0.3926 & 0.3930 \\
\cmidrule{2-12}
& MUVERA & - & 0.3678 & 0.6343 & 0.8398 & 0.2322 & 0.2875 & 0.3135 & 0.1928 & 0.2031 & 0.2040 \\
\midrule
NQ & Full MaxSim & - & 0.8428 & 0.9733 & 0.9978 & 0.6225 & 0.6533 & 0.6566 & 0.5713 & 0.5767 & 0.5768 \\
\cmidrule{2-12}
& & 8 & 0.6172 & 0.8683 & 0.9561 & 0.4159 & 0.4715 & 0.4831 & 0.3712 & 0.3802 & 0.3806 \\
& EigenLI & 16 & 0.6633 & 0.9178 & 0.9844 & 0.4676 & 0.5279 & 0.5364 & 0.4241 & 0.4353 & 0.4356 \\
& & 32 & 0.7222 & 0.9217 & 0.9794 & 0.4983 & 0.5435 & 0.5517 & 0.4467 & 0.4545 & 0.4548 \\
\cmidrule{2-12}
& & 8 & 0.4478 & 0.7194 & 0.8722 & 0.2842 & 0.3442 & 0.3638 & 0.2495 & 0.2605 & 0.2611 \\
& KMeans++ & 16 & 0.5617 & 0.8372 & 0.9328 & 0.3763 & 0.4383 & 0.4512 & 0.3357 & 0.3470 & 0.3474 \\
& & 32 & 0.6978 & 0.9117 & 0.9644 & 0.5223 & 0.5719 & 0.5790 & 0.4868 & 0.4966 & 0.4968 \\
\cmidrule{2-12}
& & 8 & 0.5633 & 0.8006 & 0.9383 & 0.3753 & 0.4289 & 0.4470 & 0.3347 & 0.3442 & 0.3448 \\
& Ward & 16 & 0.6683 & 0.9078 & 0.9694 & 0.4736 & 0.5270 & 0.5350 & 0.4291 & 0.4378 & 0.4380 \\
& & 32 & 0.7567 & 0.9578 & 0.9844 & 0.5516 & 0.5965 & 0.5999 & 0.5064 & 0.5136 & 0.5136 \\
\cmidrule{2-12}
& MUVERA & - & 0.3300 & 0.6089 & 0.8094 & 0.2010 & 0.2611 & 0.2871 & 0.1709 & 0.1829 & 0.1836 \\
\midrule
Quora & Full MaxSim & - & 0.9508 & 0.9945 & 0.9996 & 0.8610 & 0.8739 & 0.8751 & 0.8473 & 0.8486 & 0.8486 \\
\cmidrule{2-12}
& & 8 & 0.9448 & 0.9933 & 0.9995 & 0.8607 & 0.8749 & 0.8762 & 0.8512 & 0.8527 & 0.8527 \\
& EigenLI & 16 & 0.9285 & 0.9902 & 0.9991 & 0.8258 & 0.8435 & 0.8453 & 0.8087 & 0.8108 & 0.8108 \\
& & 32 & 0.8428 & 0.9609 & 0.9909 & 0.7048 & 0.7352 & 0.7401 & 0.6776 & 0.6818 & 0.6819 \\
\cmidrule{2-12}
& & 8 & 0.9358 & 0.9914 & 0.9988 & 0.8476 & 0.8636 & 0.8652 & 0.8378 & 0.8394 & 0.8394 \\
& KMeans++ & 16 & 0.9501 & 0.9943 & 0.9996 & 0.8612 & 0.8742 & 0.8753 & 0.8478 & 0.8492 & 0.8492 \\
& & 32 & 0.9508 & 0.9944 & 0.9996 & 0.8613 & 0.8743 & 0.8755 & 0.8479 & 0.8491 & 0.8491 \\
\cmidrule{2-12}
& & 8 & 0.9449 & 0.9926 & 0.9993 & 0.8603 & 0.8743 & 0.8757 & 0.8503 & 0.8517 & 0.8518 \\
& Ward & 16 & 0.9501 & 0.9944 & 0.9995 & 0.8616 & 0.8747 & 0.8759 & 0.8483 & 0.8496 & 0.8496 \\
& & 32 & 0.9508 & 0.9945 & 0.9996 & 0.8613 & 0.8743 & 0.8754 & 0.8478 & 0.8491 & 0.8491 \\
\cmidrule{2-12}
& MUVERA & - & 0.7648 & 0.9246 & 0.9797 & 0.6187 & 0.6577 & 0.6662 & 0.5922 & 0.5982 & 0.5984 \\
\midrule
SCIDOCS & Full MaxSim & - & 0.1795 & 0.4260 & 0.6835 & 0.1767 & 0.2611 & 0.3139 & 0.3124 & 0.3244 & 0.3251 \\
\cmidrule{2-12}
& & 8 & 0.1435 & 0.3700 & 0.6830 & 0.1327 & 0.2085 & 0.2725 & 0.2209 & 0.2333 & 0.2340 \\
& EigenLI & 16 & 0.1575 & 0.4020 & 0.6685 & 0.1496 & 0.2324 & 0.2868 & 0.2500 & 0.2666 & 0.2671 \\
& & 32 & 0.1635 & 0.3865 & 0.6323 & 0.1561 & 0.2321 & 0.2834 & 0.2778 & 0.2877 & 0.2888 \\
\cmidrule{2-12}
& & 8 & 0.1175 & 0.2700 & 0.6145 & 0.0985 & 0.1495 & 0.2192 & 0.1516 & 0.1626 & 0.1636 \\
& KMeans++ & 16 & 0.1490 & 0.3100 & 0.6490 & 0.1359 & 0.1895 & 0.2582 & 0.2162 & 0.2235 & 0.2247 \\
& & 32 & 0.1775 & 0.3720 & 0.6605 & 0.1643 & 0.2295 & 0.2891 & 0.2668 & 0.2746 & 0.2753 \\
\cmidrule{2-12}
& & 8 & 0.1350 & 0.3340 & 0.6610 & 0.1300 & 0.1983 & 0.2640 & 0.2252 & 0.2402 & 0.2409 \\
& Ward & 16 & 0.1655 & 0.3660 & 0.6590 & 0.1558 & 0.2244 & 0.2845 & 0.2587 & 0.2712 & 0.2719 \\
& & 32 & 0.1715 & 0.4040 & 0.6820 & 0.1652 & 0.2444 & 0.3018 & 0.2757 & 0.2867 & 0.2875 \\
\cmidrule{2-12}
& MUVERA & - & 0.1005 & 0.2455 & 0.4513 & 0.1004 & 0.1501 & 0.1927 & 0.1936 & 0.2038 & 0.2050 \\
\midrule
SciFact & Full MaxSim & - & 0.8594 & 0.9633 & 0.9967 & 0.7510 & 0.7759 & 0.7801 & 0.7260 & 0.7309 & 0.7310 \\
\cmidrule{2-12}
& & 8 & 0.8251 & 0.9433 & 0.9900 & 0.6932 & 0.7193 & 0.7255 & 0.6582 & 0.6627 & 0.6629 \\
& EigenLI & 16 & 0.8389 & 0.9450 & 0.9867 & 0.7107 & 0.7340 & 0.7394 & 0.6763 & 0.6801 & 0.6803 \\
& & 32 & 0.8424 & 0.9477 & 0.9900 & 0.7177 & 0.7409 & 0.7465 & 0.6871 & 0.6910 & 0.6912 \\
\cmidrule{2-12}
& & 8 & 0.7833 & 0.9333 & 0.9933 & 0.6189 & 0.6544 & 0.6622 & 0.5737 & 0.5803 & 0.5807 \\
& KMeans++ & 16 & 0.8398 & 0.9533 & 0.9967 & 0.6906 & 0.7153 & 0.7208 & 0.6497 & 0.6532 & 0.6534 \\
& & 32 & 0.8578 & 0.9633 & 0.9967 & 0.7294 & 0.7534 & 0.7576 & 0.6968 & 0.7007 & 0.7008 \\
\cmidrule{2-12}
& & 8 & 0.8511 & 0.9467 & 0.9933 & 0.7076 & 0.7289 & 0.7348 & 0.6688 & 0.6726 & 0.6728 \\
& Ward & 16 & 0.8478 & 0.9600 & 0.9967 & 0.7067 & 0.7321 & 0.7367 & 0.6684 & 0.6726 & 0.6727 \\
& & 32 & 0.8684 & 0.9700 & 0.9967 & 0.7482 & 0.7709 & 0.7742 & 0.7178 & 0.7216 & 0.7217 \\
\cmidrule{2-12}
& MUVERA & - & 0.7573 & 0.8777 & 0.9750 & 0.6074 & 0.6346 & 0.6468 & 0.5688 & 0.5736 & 0.5740 \\
\midrule
TREC-COVID & Full MaxSim & - & 0.8980 & 0.6830 & 0.5532 & 0.8515 & 0.6513 & 0.5812 & 0.9800 & 0.9800 & 0.9800 \\
\cmidrule{2-12}
& & 8 & 0.5320 & 0.3350 & 0.2532 & 0.5136 & 0.3366 & 0.2814 & 0.7750 & 0.7750 & 0.7750 \\
& EigenLI & 16 & 0.6280 & 0.4402 & 0.3546 & 0.5918 & 0.4298 & 0.3807 & 0.7952 & 0.7970 & 0.7970 \\
& & 32 & 0.7220 & 0.5594 & 0.4627 & 0.6593 & 0.5297 & 0.4791 & 0.8206 & 0.8206 & 0.8206 \\
\cmidrule{2-12}
& & 8 & 0.4540 & 0.2392 & 0.1602 & 0.4436 & 0.2535 & 0.1888 & 0.7937 & 0.7958 & 0.7958 \\
& KMeans++ & 16 & 0.5460 & 0.3146 & 0.2224 & 0.5203 & 0.3226 & 0.2537 & 0.7757 & 0.7772 & 0.7772 \\
& & 32 & 0.7080 & 0.4616 & 0.3436 & 0.6714 & 0.4584 & 0.3803 & 0.8850 & 0.8850 & 0.8850 \\
\cmidrule{2-12}
& & 8 & 0.5160 & 0.3054 & 0.2212 & 0.5060 & 0.3137 & 0.2513 & 0.8062 & 0.8084 & 0.8084 \\
& Ward & 16 & 0.6800 & 0.4216 & 0.3003 & 0.6415 & 0.4231 & 0.3381 & 0.8473 & 0.8473 & 0.8473 \\
& & 32 & 0.8100 & 0.5480 & 0.4086 & 0.7540 & 0.5384 & 0.4493 & 0.9012 & 0.9012 & 0.9012 \\
\cmidrule{2-12}
& MUVERA & - & 0.5560 & 0.3390 & 0.2685 & 0.5174 & 0.3379 & 0.2897 & 0.7036 & 0.7055 & 0.7055 \\
\midrule
Webis-Touche2020 & Full MaxSim & - & 0.2977 & 0.5373 & 0.8453 & 0.3159 & 0.4404 & 0.5412 & 0.5582 & 0.5648 & 0.5648 \\
\cmidrule{2-12}
& & 8 & 0.1884 & 0.4254 & 0.7893 & 0.1867 & 0.3091 & 0.4314 & 0.3638 & 0.3793 & 0.3796 \\
& EigenLI & 16 & 0.2204 & 0.4534 & 0.8023 & 0.2028 & 0.3251 & 0.4429 & 0.3287 & 0.3480 & 0.3480 \\
& & 32 & 0.2915 & 0.5059 & 0.8224 & 0.2965 & 0.4045 & 0.5117 & 0.5279 & 0.5405 & 0.5405 \\
\cmidrule{2-12}
& & 8 & 0.1373 & 0.3635 & 0.7053 & 0.1400 & 0.2504 & 0.3648 & 0.3021 & 0.3223 & 0.3225 \\
& KMeans++ & 16 & 0.1717 & 0.3784 & 0.7520 & 0.1705 & 0.2712 & 0.3945 & 0.3609 & 0.3761 & 0.3763 \\
& & 32 & 0.2256 & 0.4582 & 0.7872 & 0.2305 & 0.3437 & 0.4546 & 0.4644 & 0.4765 & 0.4765 \\
\cmidrule{2-12}
& & 8 & 0.1819 & 0.3850 & 0.7605 & 0.1707 & 0.2748 & 0.4009 & 0.3719 & 0.3820 & 0.3823 \\
& Ward & 16 & 0.2371 & 0.4373 & 0.7809 & 0.2392 & 0.3406 & 0.4567 & 0.4501 & 0.4564 & 0.4564 \\
& & 32 & 0.2979 & 0.5115 & 0.8045 & 0.3144 & 0.4180 & 0.5164 & 0.5524 & 0.5650 & 0.5650 \\
\cmidrule{2-12}
& MUVERA & - & 0.1425 & 0.2905 & 0.6168 & 0.1463 & 0.2106 & 0.3168 & 0.3841 & 0.3964 & 0.3967 \\
\midrule
\end{longtable}
\endgroup

\end{document}